\documentclass[11pt, letterpaper]{article}
\usepackage[utf8]{inputenc}
\usepackage[left=1in, top=1in, right=1in, bottom=1in]{geometry}
\usepackage[colorlinks,citecolor=blue]{hyperref}
\usepackage{amsmath, amsthm, amsfonts, amssymb}
\usepackage{mathtools}
\usepackage{bm, bbm}
\usepackage{xcolor}
\usepackage{mdframed}
\usepackage{graphicx}
\usepackage{comment}
\usepackage[ruled, linesnumbered]{algorithm2e}
\usepackage{cleveref}
\usepackage{enumitem}
\usepackage{natbib}
\usepackage{color-edits}

\theoremstyle{plain}

\newtheorem{theorem}{Theorem}
\newtheorem{claim}{Claim}
\newtheorem{lemma}{Lemma}

\theoremstyle{definition}

\usepackage{mathtools}
\usepackage{bbm}
\usepackage{natbib}
\usepackage{soul}
\usepackage{mdframed}
\mdfsetup{
 skipabove=4pt,
 skipbelow=4pt,
 leftmargin=0pt,
 rightmargin=0pt,
 innertopmargin=3pt,
 innerbottommargin=3pt,
 innerleftmargin=3pt,
 innerrightmargin=3pt,
 frametitleaboveskip=3pt,
 frametitlebelowskip=3pt,
 frametitlerule=true,
 splittopskip=2\topsep}
\usepackage{enumitem}

\usepackage{amsthm}
\usepackage{amsmath}
\usepackage{amssymb}
\usepackage{xcolor}
\usepackage{setspace}
\usepackage{booktabs} 
\usepackage[ruled, linesnumbered]{algorithm2e}

\SetAlFnt{\small}
\SetAlCapFnt{\small}
\SetAlCapNameFnt{\small}
\SetAlCapHSkip{0pt}
\IncMargin{-\parindent}

\newcommand{\eps}{\varepsilon}

\newcommand{\E}{\mathbb{E}}
\newcommand{\Reg}{\mathrm{Reg}}

\newcommand{\prior}{\mu^*}

\newcommand{\Alg}{\mathrm{Alg}}

\newcommand{\dis}{\kappa}

\newtheorem{regularity assumption}{Regularity assumption}
\newtheorem{assumption}{Assumption}

\DeclareMathOperator*{\argmax}{arg\,max}

\begin{document}

\title{Information Design with Unknown Prior\footnote{
A preliminary version of this work was published as an extended abstract at Innovations in Theoretical Computer Science (ITCS) 2025.

We are especially grateful to Bart Lipman for his detailed comments on this paper. We also thank Tilman Borgers, Yiling Chen, Xiaoyu Cheng, Krishna Dasaratha, Yiding Feng, Marina Halac, Kevin He, Michael
I. Jordan, Yingkai Li, Chiara Margaria, Teddy Mekonnen, Paul Milgrom, Jawwad Noor, Juan Ortner, Harry Pei, Pengyu Qian, Alex Smolin, Wei Tang, Allen Vong, Haifeng Xu, and Weijie Zhong for valuable discussion. This paper also benefits from audience feedback at numerous conferences. 
}}

\author{
Ce Li\thanks{Department of Economics, Boston University. \texttt{celi@bu.edu}}
\and 
Tao Lin\thanks{John A.~Paulson School of Engineering and Applied Sciences, Harvard University. \texttt{tlin@g.harvard.edu}
}
}

\maketitle

\begin{abstract}


Information designers, such as online platforms, often do not know the beliefs of their receivers. We design learning algorithms so that the information designer can learn the receivers' prior belief from their actions through repeated interactions. Our learning algorithms achieve no regret relative to the optimality for the known prior at a fast speed, achieving a tight regret bound $\Theta(\log T)$ in general and a tight regret bound $\Theta(\log \log T)$ in the important special case of binary actions.

\end{abstract}

\section{Introduction}
\label{section:intro}


An information designer (she) sends signals about the state of the world to influence the action of a receiver (he). Choosing the optimal signaling scheme requires the designer to know the prior belief of the receiver over the states.
While the designer is assumed to share a common prior belief with the receiver in canonical models (e.g., \citet{kamenica_bayesian_2011}),  the designer may know little about the prior in real-world problems. For example, consider an online dating platform that has recently entered the Boston market. The platform (designer) must decide how to present a female user’s profile (signal) to a local male user (receiver) based on the quality of their match (state). Lacking sufficient data on Boston-based users, the platform may not know the distribution of the match qualities for that male user, but the male user may know such distribution from his life experience, which is his prior belief over the matching qualities. As a result, the limited knowledge of the user's prior belief fundamentally constrains the platform's ability to tailor an optimal signaling scheme.

Motivated by such scenarios, {we study an information design problem where the prior belief of the receivers, assumed to be the true distribution over the states, is unknown to the designer.} 
{Standard economic models would assume that the designer holds a belief over the unknown prior belief} (e.g., \cite{alonso_bayesian_2016, kolotilin_persuasion_2017}). But it is computationally infeasible (NP hard) to compute the optimal signaling scheme in that case 
\citep{hossain_multi-sender_2024}.
On the other hand, recent robustness theories consider a designer optimizing her signaling scheme for the worst-case prior from a set of candidates (e.g., \cite{hu_robust_2021, Piotr_2022}).
But the resulting signaling scheme could be arbitrarily bad for the designer when the actual prior is not the worst-case one. 


Can the designer still design near-optimal signaling schemes when the prior is unknown?
We find that the designer is able to design approximately optimal signaling schemes by \textit{learning the prior} through repeated interactions with the receivers.
More importantly, our designer achieves such approximate optimality at an \textit{efficient} speed. 

{Specifically, we study a \emph{learning} model where an information designer (e.g., the online dating platform) learns to design signaling schemes over $T$ periods based on historical interactions with a sequence of one-period-lived (or, myopic)
receivers (e.g., male users).} The prior beliefs of the receivers are the same as the true distribution of the states, but the designer does not know that distribution.
In each period, a new receiver arrives, 
and the designer commits to a signaling scheme, which maps a state i.i.d.~sampled from the prior to a randomized signal.
Upon receiving the signal, the receiver in that period updates his belief about the state and takes a best-response action accordingly. The payoffs of the two players are then realized.
%
%
%
We aim to design a learning algorithm for the designer to improve the signaling schemes over time, such that the designer's time-average payoff converges to the payoff of the optimal signaling scheme for the unknown prior as time $T$ approaches infinity, and moreover, converges at a fast speed.


The rate at which the designer's payoff converges to the optimal objective depends on the learning algorithm she employs.
We measure the performance of the learning algorithm by the notion of \emph{regret}:
the time-cumulative difference between the optimal payoff and the designer’s actual payoff.
A learning algorithm is said to be \textit{no-regret} if its regret grows sublinearly with time (alternatively, its time-averaged regret converges to zero).
The rate at which the time-averaged regret declines reflects the algorithm’s efficiency.
We aim to design learning algorithms that not only achieve no regret, but do so as efficiently as possible, yielding fast convergence to the optimal payoff under the unknown prior.

{
Our first main result is an $O(\log T)$ upper bound on the regret: we design a learning algorithm to ensure that, for any problem instance, the designer's regret for $T$ periods is no more than a constant times $\log(T)$ for all $T$.}
Since the prior belief of the receivers is the true distribution of the states, a natural approach for the designer to learn the prior is to use \emph{the empirical state distribution} (e.g., \citet{zu_learning_2021}). However, the empirical estimation approach has two key drawbacks: (1) it requires observing the states, which is infeasible if the states are hidden to the designer; (2) it suffers at least $O(\sqrt{T})$ regret
due to its estimation error. We depart from empirical estimation by learning the prior \emph{from the receiver's actions}. 

Our learning algorithm is based on \textit{binary search}. The key idea is to learn the unknown prior from the receivers' actions.
Consider the special case of two states and two actions, where the receiver prefers different actions for the two states.
The designer can ``guess a prior'' by designing a signaling scheme such that the receiver is indifferent between the two actions under the posterior updated from the guessed prior. 
The actual action chosen by the receiver then reveals whether the receiver's true prior is below or above the designer's guess. 
Based on that, the designer gradually narrows down the range of possible values for the prior. That process continues until the designer obtains a sufficiently accurate estimation of the prior.
Our formal analysis shows how to generalize that intuition to the case with multiple states and multiple actions. Then, we compute an optimal signaling scheme for the estimated prior.
However, even if two priors are very close, the optimal signaling scheme for one prior might not work well for the other prior, leading to unexpected influences on the designer's regret.
We thus \emph{robustify} the signaling scheme that is optimal for the learned prior to ensure that it \textit{also} works well for the true prior. 

The superiority of our learning algorithm is reflected by the upper bound of the regret it incur: $O (\log T)$. That suggests that the designer's average payoff converges to the optimality at a fast speed, in the order of $O (\frac{\log T}{T})$, which is significantly faster than the intuitive approach, empirical estimation, which is only in the order of $O (\frac{\sqrt T}{T})$. 
Our binary-search-based algorithm allows the designer to efficiently narrow down the possible range of the unknown prior to a set of size $\eps = O(\frac{1}{T})$ in $O(\log T)$ rounds.
Since the exploration phase only takes $O(\log T)$ rounds, and the regret incurred per round is bounded, the total regret is $O(\log T)$. 


{
In parallel with our first main result of $O(\log T)$ regret upper bound, we also provide a lower bound result.  We prove that the designer's regret cannot grow more slowly than the order of $\log(T)$.
More formally, no matter what learning algorithm the designer employs, there always exists a problem instance such that the designer's regret for $T$ periods is at least a constant times $\log(T)$.
This lower bound result implies that no algorithm can achieve a smaller regret than $\log T$, so our binary-search-based algorithm is optimal up to a constant factor. The matching upper and lower bounds on regret exactly pin down the fastest-possible rate of convergence of the designer's average payoff: it is in the order of $O(\frac{\log T}{T})$. 
}


{Having shown that the tight regret bound for the designer is $O(\log T)$ in general, we shift to a special yet practically relevant case where the receiver has two actions.} 
For example, a male user on the online dating platform may choose to start a conversation with a recommended female or not; an investor on an online financial platform decides whether to invest in the recommended asset or not. In such scenarios, the designer often prefers one action over the other: the dating platform always prefers the user to start a new conversation since that makes the user more willing to use the platform; the financial platform wants the investor to invest regardless of the state, since the platform charges a transaction fee for each investment. 
{For such binary-action scenarios, we prove that the designer's regret can surprisingly be reduced to $O(\log \log T)$, which is an exponential improvement over the general case.} 
To key idea is that the optimal signaling scheme in the binary-action case depends on a single parameter that we call \textit{persuasion strength}, which determines how often to recommend the designer-preferred action. Though unknown, the optimal persuasion strength can be approximated efficiently by testing a sequence of signaling schemes. The designer first performs a coarse search to narrow down the range, then refines the estimate with smaller steps. At each stage, the designer checks if the receiver finds the recommendation persuasive. This two-phase search significantly reduces the exploration cost, resulting in total regret of $O(\log \log T)$, after which the designer uses a near-optimal signaling scheme.
We also provide a regret lower bound in the order of $\log \log T$, which certifies the optimality of our learning algorithm. 



\subsection{Related Literature}
\label{sec:related-literature}

This paper contributes to the growing literature on robust information design under uncertainty, particularly in the problem where the designer lacks knowledge of the receiver’s prior belief.
Classical models such as Bayesian persuasion \citep{kamenica_bayesian_2011, BergemannMorris2016} and cheap talk \citep{crawford_strategic_1982} typically assume that the designer and receiver share a common prior. 
Recent works have relaxed this assumption by introducing uncertainty over the prior. One approach assumes that the designer holds a belief over the unknown prior (e.g., \cite{kolotilin_persuasion_2017}), but computing the optimal signaling scheme under this model is generally intractable \citep{hossain_multi-sender_2024}. Another approach, rooted in robust information design, has the designer optimize against the worst-case prior in a given set (e.g., \cite{hu_robust_2021, Piotr_2022, Kosterina2022, Dworczak_Kolotilin_2024}). While conceptually appealing, this worst-case strategy can perform poorly when the prior is not the one considered in the worst case. These previous works are based on a single-shot game, limiting the learning of the designer over time. In contrast, the present paper designs learning algorithms for the designer to learn the prior, enabling the designer to achieve near optimality at a fast speed 
and in a computationally efficient way.

There are also papers that use other approaches to address the robustness. For example, \cite{Mathevet2020} and \cite{morris2024} study the designer's strategy profile from the adversarial perspective in a one-shot game. \cite{li_norman_2021} study how to choose the designer's strategy from the adversarial perspective in sequential games, but their paper keeps the common prior assumption. In these papers, it is either that the common prior is assumed or that the model is a one-shot game, while our paper tackles a learning process and the unknown prior belief at the same time.

Our work also contributes to the rich literature on {online Bayesian persuasion} that studies the learning problem of an information designer who does not know some parameters of the game. \cite{castiglioni_online_2020, castiglioni_multi-receiver_2021, feng_online_2022} consider unknown utility functions of the receivers, \cite{zu_learning_2021, wu_learning_2022, harris2023algorithmic} focus on unknown prior, while in \cite{bacchiocchi_markov_2024} both the receiver's utility and the prior are unknown.  The approach to tackling unknown prior in most previous papers is empirical estimation, which suffers an $\Omega(\sqrt T)$ regret due to sampling error.  We bypass this negative result by using the receiver's best-responding action to infer the prior much more efficiently, thus achieving an $O(\log T)$ regret. 
\citet{harris2023algorithmic} also learn the receiver's prior from best-responding actions. They aim to characterize the exactly optimal learning algorithm in a specific model with real-value state and binary actions, while we focus on achieving logarithmic regret (up to a constant factor). With a less ambitious goal, our results apply to more general models with non-binary actions. 

\cite{Camara2020MechanismsFA} study a problem where both information designers and receivers are learning agents. Different from our work, \cite{Camara2020MechanismsFA} completely drop the distributional assumption by assuming an adversarial sequence of states of the world. Our work assumes that states are i.i.d.~drawn from an unknown objective prior (i.e., the true distribution). Thus, their definition of regret is different from ours.  

Finally, our work studies the convergence of the regret and emphasizes how quickly such convergences occur, which is also related to the line of literature on the rate of convergence in economic theory, such as the convergence of the amount of information (\cite{LIANG_MU_2020, frick2024monitoringrichdata}) and the value of attributes (\cite{liang2025}). However, our paper considers the convergence of regret, which is the performance measure for the learning algorithms. Essentially, the regret convergence in our paper is up to our design for the learning algorithms instead of subjective beliefs or exogenous information sources.

\section{Models and Preliminaries}
\label{sec:preliminary}


\subsection{Models}
\label{sec:prelim-BP}

\paragraph*{Bayesian Persuasion}
The single-period game of our learning model is the Bayesian persuasion model of \cite{kamenica_bayesian_2011}. 
It is a game between an \emph{information designer} (she) and a \emph{receiver} (he). There is a finite set of payoff-relevant states of the world $\Omega$, where the states follow a distribution $\prior \in \Delta(\Omega)$. 
The receiver has a finite set of actions $A$. 
At the start of the game, the designer designs and announces a signaling scheme $\pi: \Omega \rightarrow \Delta(S)$, which is a mapping from each state to a probability distribution over signals in some finite signal set $S$.  We denote by $\pi(s|\omega)$ the probability that signal $s$ is realized conditional on state $\omega$. 
Then, a state $\omega$ is drawn from the distribution $\prior$, and a signal $s$ is generated according to $\pi(\cdot | \omega)$. The receiver does not observe the state $ \omega$, but instead observes the realized signal $s$.

Similarly to \cite{kamenica_bayesian_2011}, we assume that the signal can be generated by a third party (e.g., Nature), so the designer may not observe the state $\omega$.  
The designer observes the signal $s$ as the receiver does.
Upon receiving signal $s$, the receiver, who knows $\prior$ \emph{as his prior belief} and the signaling scheme $\pi$, updates his belief about the state to a posterior $\mu_{s, \pi}$ using Bayes' rule:
\begin{equation}
    \mu_{s, \pi}(\omega) ~ = ~ \tfrac{\mu^*(\omega)\pi(s|\omega)}{\sum_{\omega_i \in \Omega} \mu^*(\omega_i)\pi(s|\omega_i)}, \quad \forall \omega \in \Omega, 
\end{equation}
and then takes an optimal action with respect to $\mu_{s, \pi}$: 
\begin{align}
    a^*_{s, \pi} ~ \in ~ \argmax_{a\in A} \E_{\omega\sim \mu_{s, \pi}} [v(a, \omega)] ~ = ~ \argmax_{a\in A} \sum_{\omega \in \Omega} \prior(\omega) \pi(s|\omega) v(a, \omega), 
\end{align}
where $v : A \times \Omega \to [0, 1]$ is a bounded utility function of the receiver.  The designer then obtains utility $u(a^*_{s, \pi}, \omega)$
where $u : A\times \Omega \to [0, 1]$ is the designer's bounded utility function.  The expected utility of the designer, as a function of prior $\prior$ and signaling scheme $\pi$, is thus
\begin{equation}
U(\prior, \pi) ~ = ~ \E\big[ u(a^*_{s, \pi}, \omega) ] ~ = ~ \sum_{\omega \in \Omega} \prior(\omega) \sum_{s\in S} \pi(s|\omega) u(a^*_{s, \pi}, \omega). 
\end{equation}
The goal of the designer is to find a signaling scheme $\pi^*$ that maximizes her expected utility: 
\begin{equation}
    \pi^* ~ \in ~ \argmax_\pi U(\prior, \pi). 
\end{equation}
We call such a $\pi^*$ an optimal signaling scheme for prior $\prior$ and denote its utility by $U^* = U(\prior, \pi^*)$.  A signaling scheme $\pi$ is \emph{$\eps$-approximately optimal} (or \emph{$\eps$-optimal}) if 
\begin{equation}
    U(\prior, \pi) ~ \ge ~ U^* - \eps ~ = ~ \max_\pi U(\prior, \pi) - \eps.
\end{equation}
The \emph{single-period regret}
of a signaling scheme $\pi$ is the difference between the optimal utility for prior $\mu^*$ and the utility of $\pi$: $U^* - U(\prior, \pi)$. {We summarize the single-period Bayesian persuasion game by an \textit{instance} $\mathcal{I}=\{\Omega, A, u, v, \mu^*\}$. }

\paragraph*{Direct and Persuasive Signaling Schemes}
It is well known that, without loss of generality, the optimal signaling scheme $\pi^*$ in Bayesian persuasion can be assumed to be \emph{direct} and \emph{persuasive} \citep{kamenica_bayesian_2011}. 
A \emph{direct} signaling scheme is a signaling scheme $\pi: \Omega \rightarrow \Delta(A)$ that maps every state to a probability distribution over actions, so every signal $a \in A$ is an action recommendation for the receiver.  
A direct signaling scheme $\pi$ is \textit{persuasive for action/signal $a \in A$} if the recommended action $a$ is optimal for the receiver: $\sum_{\omega} \mu(\omega) \pi(a|\omega) \big[ v(a,\omega)-v(a',\omega) \big] \ge 0, \forall  a' \in A$; 
and $\pi$ is called \textit{persuasive} if it is persuasive for all actions $a \in A$.  We denote by $\mathrm{Pers}(\mu)$ the set of all persuasive signaling schemes when the receiver has prior $\mu$:
\begin{equation}
    \mathrm{Pers}(\mu) \coloneqq \Big\{ \pi: \Omega \to \Delta(A) ~ \big | ~  \sum_{\omega} \mu(\omega) \pi(a|\omega) \big[ v(a,\omega)-v(a',\omega) \big] \ge 0, ~ \forall a,a' \in A \Big\}.
\end{equation}
With attention restricted to persuasive signaling schemes, the optimal signaling scheme $\pi^*$ for a given prior $\prior$ can be computed efficiently by a linear program \citep{Dughmi_ABP}.

\paragraph*{$T$-Period Learning Model}
We consider a $T$-period interaction between an information designer (e.g., an online dating platform) and a sequence of receivers (e.g., male users), {where the single-period game $\mathcal I = \{\Omega, A, u, v, \prior\}$ is repeated $T$ times.} The prior belief of each receiver coincides with the true state distribution $\prior$. Importantly, $\prior$ is unknown to the designer. The designer knows her own utility function $u$ as well as each receiver’s utility function $v$.

In each period $t \in \{1, \ldots, T\}$, a new receiver arrives and the designer announces a signaling scheme $\pi^{(t)}$, which is designed by an algorithm that learns the prior $\mu^*$ from historical information, including past signaling schemes, realized signals, the actions taken by the receivers, and their payoff functions $u,v$.
A new state $\omega^{(t)} \sim \mu^*$ is independently sampled, and then Nature sends a signal $s^{(t)} \sim \pi^{(t)}(\cdot | \omega^{(t)})$ to the receiver. The receiver then performs a Bayesian update based on $\mu^*$, $\pi^{(t)}$, and $s^{(t)}$, and takes an optimal action $a^{(t)}$ (breaking ties arbitrarily) for the posterior belief $\mu_{s^{(t)}, \pi^{(t)}}$.

The \textit{regret} of the designer {under instance $\mathcal{I} = \{\Omega, A, u, v, \prior\}$} is the difference between the optimal utility and the actual (expected) utility that the designer obtains by using the algorithm over the $T$ periods, that is, 
\begin{equation*}
    \mathrm{Reg}(T {; \mathcal{I}}) \, \coloneqq \, T \cdot U^* - \E\Big[ \sum_{t=1}^T u(a^{(t)}, \omega^{(t)}) \Big] \, = \, T \cdot U^* - \E_{\pi^{(1)}, \ldots, \pi^{(T)}}\Big[ \sum_{t=1}^T U(\prior, \pi^{(t)}) \Big].
\end{equation*}
Clearly, \emph{if} the designer is able to use an $\eps$-approximately optimal signaling scheme $\pi$ for prior $\mu^*$ in all $T$ periods, then her regret will be at most $\eps T$.  
Thus, the goal of the designer is to find approximately optimal signaling schemes while learning the unknown prior $\mu^*$ during the $T$ periods. 

{Our designer may or may not observe the realized state at each period. Regardless of whether the designer can observe the state, our learning algorithm will ensure that the designer's time-averaged regret converges to zero at a much faster rate compared to simply using state observations to estimate the distribution. 
}



A regret \textit{upper bound} is a performance guarantee of an algorithm in the following form: {for any instance $\mathcal I$, that is, for any unknown prior belief $\prior$,}
the regret incurred by the algorithm grows no faster than some function of $T$, i.e.,
{
\begin{equation*}
    \Reg(T; \mathcal I) ~ \leq ~ C_{\mathcal{I}} \cdot f(T) ~ = ~ O(f(T)), \quad \forall T, ~ \forall \mathcal I,
\end{equation*}
where $C_{\mathcal{I}}$ is a constant depending on $\mathcal I$ but not on $T$. }
Here, $f(T)$ represents the \emph{worst-case} rate at which regret can grow. For example, the empirical estimation has $\Reg(T; \mathcal I) = O(\sqrt{T})$. An algorithm with regret upper bound $O(f(T))$ is \textit{no-regret} if $f(T)$ is sub-linear in $T$: $\frac{f(T)}{T}\to 0$.  

A regret \textit{lower bound} means that \textit{for any algorithm,} {
there exists some instance $\mathcal I$, i.e., for some unknown prior $\prior$,}
such that
{
\begin{equation*}
    \Reg(T; \mathcal I) ~ \ge ~ c_{\mathcal I} \cdot g(T) ~ = ~ \Omega( g(T)), \quad \forall T,
\end{equation*}
for some instance-dependent constant $c_{\mathcal I}>0$.}
The function $g(T)$ captures the unavoidable cost of not knowing $\prior$ no matter what learning algorithm is used, 
which fundamentally shows how hard the problem of information design with unknown prior is. 

{
When the upper and lower bounds match in order, i.e.,
\begin{equation*}
    \Omega( h(T) ) ~ \le ~ \Reg(T; \mathcal I) ~ \le ~ O(h(T))
\end{equation*}
with some function $h(T)$, we say the regret bound is \emph{tight} and denote it by $\Reg(T; \mathcal I)=\Theta(h(T))$. }

\subsection{Robustification of Signaling Schemes}
A technique that we will use to design signaling schemes when the designer does not know the prior belief of the receiver is \emph{robustification}.  Suppose we obtained an estimate $\hat \mu$ of the unknown prior $\prior$ with $\ell_1$-distance $\| \hat \mu - \mu^* \|_1 \le \eps$.  We can compute a signaling scheme $\hat \pi$ that is optimal for the estimate $\hat \mu$, which is thus persuasive for $\hat \mu$.
However, the signaling scheme $\hat \pi$ might be non-persuasive for the prior $\prior$. What is even worse, the signaling scheme $\hat \pi$ may perform poorly on $\prior$ because, even though a signal $s$ induces two similar posterior distributions $\hat \mu_{s, \hat \pi}$ and $\prior_{s, \hat \pi}$ under two close priors $\hat \mu$ and $\prior$, the best-responding actions of the receiver under the two posteriors might be very different, leading to an arbitrarily low payoff for the designer under the prior $\prior$. 

The idea of robustification is to slightly modify the signaling scheme $\hat \pi$ to be another scheme $\pi$ that is persuasive and approximately optimal for all priors that are close to $\hat \mu$, including $\prior$.  Robustification requires some mild and standard assumptions on the prior and the receiver's utility function (e.g., \cite{zu_learning_2021}): 
\begin{assumption}[Regularity of prior] \label{regularity2} \label{ass:prior-p0}
The prior $\prior \in \Delta(\Omega)$ has full support {and there exists $p_0 > 0$ such that} the designer knows $\min_{\omega\in\Omega} \prior(\omega) \ge p_0 > 0$.
\end{assumption}


\begin{assumption}[Regularity of receiver's utility] \label{ass:D}
There is no weakly dominated action for the receiver.  
\end{assumption}

Assumption \ref{ass:D} implies that there exists a positive number $D > 0$
such that, for every action $a \in A$ of the receiver, there exists a belief $\eta_a \in \Delta(\Omega)$ on which action $a$ is \textit{strictly} better than any other action by a margin $D$:
\begin{equation}
\label{eq:D}
    \E_{\omega \sim \eta_a}[v(a, \omega)] ~\ge ~ \E_{\omega \sim \eta_a}[v(a', \omega)] + D, \quad \forall a' \in A\setminus\{a\}.
\end{equation}
The designer knows $D$ since she knows the receiver's utility function $v$.



The following lemma formalizes the idea of robustification.

\begin{lemma}[Robustification]\label{lem:robust-pi}
Assume~\ref{ass:prior-p0} and \ref{ass:D}. Suppose $\|\hat \mu - \prior \|_1 \le \eps \le \frac{p_0^2 D}{2}$.  Any persuasive signaling scheme $\hat \pi$ for prior $\hat \mu$ can be converted into into another direct signaling scheme $\tilde \pi$ such that
\begin{itemize}
    \item $\tilde \pi$ is persuasive for any prior $\prior$ in $B_1(\hat \mu, \eps) = \{ \mu \in \Delta(\Omega): \| \mu - \hat \mu \|_1 \le \eps\}$; 
    \item the designer's utility $U(\hat \mu, \tilde \pi) \ge U(\hat \mu, \hat \pi) - \tfrac{6\eps}{p_0^2 D}$;
    \item as a corollary, if $\hat \pi$ is optimal for $\hat \mu$, then $\tilde \pi$ is $\tfrac{14\eps}{p_0^2 D}$-optimal for $\prior$.  
\end{itemize}

\end{lemma}

{The high-level idea of robustification is as follows. To robustify a persuasive signaling scheme $\hat \pi$ for the estimate $\hat \mu$, we first gently tilt each recommendation $a$'s induced belief toward belief $\eta_a$ where that action is strictly optimal. We then add to that a small probability of full revelation so the average belief satisfies Bayes plausibility.
Finally, we coalesce those ``rare" revelation messages back into direct action recommendations. Choosing the tilt size proportional to the precision parameter $\eps$ makes every recommendation $a$ persuasive for all nearby candidate beliefs. The designer's utility loss is proportional to that tilt in the order of $\eps$. See Appendix \ref{app:robust-pi} for details.}

{A similar robustification idea appeared in the work by \citet{zu_learning_2021}. However, their robustification only works for the optimal signaling scheme for the estimated prior. Our robustification does not require the signaling scheme to be optimal and works for any persuasive signaling scheme for the estimated prior.}




\section{Fast Learning to Persuade Receivers: General Case}
\label{sec:prior-aware}



This section studies the $T$-period learning model of information design with unknown prior, assuming that the designer and receivers have general utility functions. 
We will design a learning algorithm for the designer to achieve a regret upper bounded by $O(\log T)$. 
We will also show that the regret incurred by any learning algorithm is lower bounded by $\Omega(\log T)$ in the worst case.
Thus, the learning algorithm we design is not only efficient but also optimal up to a constant factor.

\subsection{$O(\log T)$ Regret Upper Bound}
\label{subsec:general-upper-bound}

At a high level, our learning algorithm has two phases: \emph{exploration} and \emph{exploitation}. During exploration, the designer learns the prior $\prior$ from the receivers' actions by using 
a binary-search-based algorithm that we develop. During exploitation, we apply robustification to the optimal signaling scheme for the prior estimate to obtain an approximately optimal signaling scheme for the prior $\prior$.

\paragraph*{Efficient learning of prior.} 

How do we learn the unknown prior $\mu^*$? 
Instead of estimating it using the empirical distribution of the state of the world, we design a much more efficient algorithm that uses the actions taken by the receivers to infer the prior $\prior$.  Recall that a receiver takes an action that is optimal for his posterior belief obtained by Bayes updating the prior $\prior$ after receiving a signal. Such an action contains information about the prior $\prior$. By employing multiple different signaling schemes, the designer can gradually acquire information to learn the prior accurately. Such a process requires a natural and mild assumption on the receiver's utility function $v(\cdot, \cdot)$: 
\begin{assumption}[Unique optimal action] \label{ass:receiver-optimal-action}
For each state $\omega \in \Omega$, the optimal action $a_\omega = \argmax_{a \in A} v(a, \omega)$ for the receiver is unique and strictly better than any other action by a positive margin of $G$: $v(a_\omega, \omega) - v(a', \omega) > G > 0, \forall a' \in A\setminus \{ a_\omega\}$. 
\end{assumption}



Specifically, we learn the prior $\prior \in \Delta(\Omega)$ by estimating the ratio of probability $\frac{\prior(\omega_i)}{\prior(\omega_1)}$ between every state $\omega_i$ ($i=2, \ldots, |\Omega|$) and a fixed state $\omega_1$. How do we estimate the probability ratio between two states?
We first show how to do this for a pair of states, say $\omega_1, \omega_2$, whose receiver-optimal actions $a^*_{\omega_1}, a^*_{\omega_2}$ are different.  We call such a pair of states \emph{distinguishable}, and assume that such a pair exists:

\begin{assumption}[Distinguishable states] \label{ass:two-states-different-actions}
There exist two states $\omega_1, \omega_2 \in \Omega$ whose corresponding receiver-optimal actions are different: $a^*_{\omega_1} \ne a^*_{\omega_2}$.
\end{assumption}




{Assumption \ref{ass:two-states-different-actions} makes the learning problem of the designer non-trivial. Otherwise, the receiver has the same optimal action $a^*$ in all states. Thus, regardless of the signaling scheme and the prior, the receiver will find the action $a^*$ optimal at his posterior. 
The designer will have a constant utility 
$\sum_{\omega \in \Omega} \prior(\omega) u(a^*, \omega)$, thereby getting zero regret regardless of the the signaling scheme.} 

Algorithm~\ref{alg:binary-search} shows how to estimate the prior probability ratio $\frac{\mu^*(\omega_1)}{\mu^*(\omega_2)}$ between a pair of distinguishable states $\omega_1$ and $\omega_2$ using multiple signaling schemes.  The main idea is a binary search. Specifically, if the designer uses a signaling scheme $\pi$ that sends a certain signal $s_0$ only under states $\omega_1$ and $\omega_2$ (namely, $\pi(s_0 | \omega) = 0$ for $\omega \ne \omega_1, \omega_2$), then the receiver will believe that the state must be $\omega_1$ or $\omega_2$ whenever he receives signal $s_0$. If the signaling ratio $\frac{\pi(s_0 | \omega_2)}{\pi(s_0 | \omega_1)}$ is zero, then the receiver will know for sure that the state is $\omega_1$, thereby taking the optimal action $a_1$ for state $\omega_1$.  On the other hand, if the signaling ratio $\frac{\pi(s_0 | \omega_2)}{\pi(s_0 | \omega_1)}$ is large (say, larger than $\frac{1}{Gp_0}$), then the receiver will believe that the state is $\omega_2$ with a high probability and take the optimal action for $\omega_2$, which is different from $a_1$ by Assumption~\ref{ass:two-states-different-actions}. Such reasoning suggests that there exists some threshold $\tau \in [0, \frac{1}{Gp_0}]$ such that the receiver will start taking some action $\tilde a \ne a_1$ when the signaling ratio $\frac{\pi(s_0 | \omega_2)}{\pi(s_0 | \omega_1)}$ is above $\tau$. 
The receiver must be indifferent between taking action $a_1$ and the different action $\tilde a$ at the signaling threshold, i.e., at $\frac{\pi(s_0 | \omega_2)}{\pi(s_0 | \omega_1)}=\tau$, we have
\begin{align*}
    \mu^*(\omega_1) \pi(s_0 | \omega_1) \big[ v(a_1, \omega_1) - v(\tilde a, \omega_1) \big] + \mu^*(\omega_2) \pi(s_0 | \omega_2) \big[ v(a_1, \omega_2) - v(\tilde a, \omega_2) \big] = 0.
\end{align*}
That implies
\begin{align*}
    \frac{\mu^*(\omega_1)}{\mu^*(\omega_2)} = \frac{\pi(s_0 | \omega_2)}{\pi(s_0 | \omega_1)} \cdot \frac{v(\tilde a, \omega_2) - v(a_1, \omega_2)}{v(a_1, \omega_1) - v(\tilde a, \omega_1)} = \tau \cdot \frac{v(\tilde a, \omega_2) - v(a_1, \omega_2)}{v(a_1, \omega_1) - v(\tilde a, \omega_1)},
\end{align*}
which gives the value of the prior ratio $\frac{\mu^*(\omega_1)}{\mu^*(\omega_2)}$.

\begin{algorithm}[H]
\caption{Estimating Prior Ratio for Distinguishable States by Binary Search}
\label{alg:binary-search}

\SetKwInOut{Input}{Input}
\SetKwInOut{Output}{Output}
\SetKwInOut{Parameter}{Parameter}

\Input {two states $\omega_1, \omega_2$ whose receiver-optimal actions are different} 
\Parameter {a desired accuracy $\eps > 0$}
\Output {an estimation $\hat \rho$ of the ratio $\frac{\mu^*(\omega_1)}{\mu^*(\omega_2)}$ that satisfies $|\hat \rho - \frac{\mu^*(\omega_1)}{\mu^*(\omega_2)}| \le \eps$}

\DontPrintSemicolon
\LinesNumbered

Let $k = 0$, $\ell^{(0)} = 0$, $r^{(0)} = \frac{1}{G p_0}$. \;
Let $a_1 = \argmax_{a\in A} v(a, \omega_1)$, $\tilde a = \argmax_{a \in A} v(a, \omega_2)$. ($a_1 \ne \tilde a$ by assumption) \;
\While{$r^{(k)} - \ell^{(k)} > \eps G$} {
    Let $q = \frac{\ell^{(k)} + r^{(k)}}{2}$. \;
    Let $s_0$ be an arbitrary signal in $S$; let $\pi^{(k)}$ be a signaling scheme that satisfies $\frac{\pi^{(k)}(s_0 | \omega_2)}{\pi^{(k)}(s_0 | \omega_1)} = q$ and $\pi^{(k)}(s_0|\omega) = 0$ for $\omega \ne \omega_1, \omega_2$ (see the proof of Lemma~\ref{lem:ratio} for how to construct such a $\pi^{(k)}$).  \;
    Use $\pi^{(k)}$ for multiple periods until signal $s_0$ is sent.  Let $a^{(k)}$ be the action taken by the receiver when $s_0$ is sent. \;
    \eIf {$a^{(k)} = a_1$} {
        Let $\ell^{(k+1)} = q, r^{(k+1)} = r^{(k)}$.
    } {
        Let $\ell^{(k+1)} = \ell^{(k)}, r^{(k+1)} = q$, $\tilde a = a^{(k)}$. 
    }
    $k = k+1$. \;
}
Output $\hat \rho = \ell^{(k)} \cdot \frac{v(\tilde a, \omega_2) - v(a_1, \omega_2)}{v(a_1, \omega_1) - v(\tilde a, \omega_1)}$. 
\end{algorithm}
\vspace{0.2em}

The performance of Algorithm \ref{alg:binary-search} is characterized by Lemma \ref{lem:ratio}, which shows that a good estimate of $\frac{\mu^*(\omega_1)}{\mu^*(\omega_2)}$ can be obtained within a very short amount of time with high probability.

\begin{lemma} 
\label{lem:ratio}
The output $\hat \rho$ of Algorithm \ref{alg:binary-search} satisfies $\hat \rho \le \frac{\mu^*(\omega_1)}{\mu^*(\omega_2)} \le \hat \rho + \eps$, and Algorithm~\ref{alg:binary-search} terminates in at most $\frac{1}{p_0} \log_2 \frac{1}{G^2p_0 \eps}$ periods in expectation. 
\end{lemma}
\begin{proof}
    See Appendix \ref{proof:ratio}. 
\end{proof}


What if the receiver-optimal actions of a pair of states are the same? Algorithm \ref{alg:any-pair-of-states} deals with such a case, allowing the designer to estimate the prior ratio between any pair of states.
The idea of Algorithm \ref{alg:any-pair-of-states} is that, if states $\omega_i$ and $\omega_j$ have the same receiver-optimal actions, then both of them must be distinguishable from one state in a distinguishable pair of states, say $\omega_k \in \{\omega_1, \omega_2\}$. Thus, we estimate the ratios $\frac{\prior(\omega_i)}{\prior(\omega_k)}$ and $\frac{\prior(\omega_j)}{\prior(\omega_k)}$ separately, which will give the probability ratio $\frac{\prior(\omega_i)}{\prior(\omega_j)}$ for the two indistinguishable states by a division. 

\begin{algorithm}[H]
\caption{Estimating Prior Ratio for Any Pair of States}
\label{alg:any-pair-of-states}

\SetKwInOut{Input}{Input}
\SetKwInOut{Output}{Output}
\SetKwInOut{Parameter}{Parameter}

\Input {any two states $\omega_i, \omega_j \in \Omega$} 
\Parameter {accuracy $\eps > 0$}
\Output {an estimation $\hat \rho_{ij}$ of the ratio $\frac{\mu^*(\omega_i)}{\mu^*(\omega_j)}$}

\DontPrintSemicolon
\LinesNumbered

If $\omega_i$ and $\omega_j$ are distinguishable, i.e., $a^*_{\omega_i} \ne a^*_{\omega_j}$, then run Algorithm~\ref{alg:binary-search} on $\omega_i$ and $\omega_j$ with parameter $\eps$. \;

Otherwise, i.e., $a^*_{\omega_i} = a^*_{\omega_j}$, find $\omega_k \in \{\omega_1, \omega_2\}$ such that $a_{\omega_k}^* \ne a^*_{\omega_i} = a^*_{\omega_j}$.  Run Algorithm~\ref{alg:binary-search} with parameter $\eps$ to obtain an estimate $\hat \rho_{ik}$ for $\frac{\prior(\omega_i)}{\prior(\omega_k)}$ and an estimate $\hat \rho_{jk}$ for $\frac{\prior(\omega_j)}{\prior(\omega_k)}$.  Return $\hat\rho_{ij} = \frac{\hat \rho_{ik}}{\hat \rho_{jk}}$. 
\end{algorithm}

\begin{lemma}
\label{lem:relax-ratio}
Suppose $\eps \le \frac{p_0}{2}$. 
For any two states $\omega_i, \omega_j \in \Omega$, the output $\hat \rho_{ij}$ of Algorithm \ref{alg:any-pair-of-states} satisfies $|\hat \rho_{ij} - \frac{\prior(\omega_i)}{\prior(\omega_j)}| \le \frac{2\eps}{p_0^2}$.  Algorithm~\ref{alg:any-pair-of-states} terminates in at most $\frac{2}{p_0} \log_2 \frac{1}{G^2p_0 \eps}$ periods in expectation. 
\end{lemma}
\begin{proof}
The proof uses Lemma~\ref{lem:ratio} twice.  See details in Appendix \ref{proof:relax-ratio}. 
\end{proof}



\paragraph*{Full algorithm.}

After showing how to estimate the prior ratio between any pair of states using Algorithm~\ref{alg:any-pair-of-states}, we then present the full algorithm, Algorithm \ref{alg:known-prior}, for the designer. In Algorithm \ref{alg:known-prior}, we first estimate the ratio $\frac{\mu^*(\omega_i)}{\mu^*(\omega_1)}$ for every state $\omega_i$, $i=2, \ldots, |\Omega|$. That allows us to reconstruct the unknown prior $\mu^*$ with a high precision, i.e., obtaining an estimation $\hat \mu$ satisfying $\| \hat \mu - \mu^* \|_1 \le O(\eps)$.  Then, we use the robustification technique (Lemma~\ref{lem:robust-pi}) to obtain a signaling scheme $\pi$ that is persuasive and $O(\eps)$-approximately optimal for $\mu^*$. 
Using the robustified signaling scheme $\pi$ for the remaining periods thus incurs a small regret. The total regret of Algorithm \ref{alg:known-prior} is formally characterized in Theorem \ref{thm:prior-aware}.

\vspace{0.5em}
\begin{algorithm}[H]
\caption{
``Learning + Robustification'' for Persuasion with Unknown Prior
}
\label{alg:known-prior}

\SetKwInOut{Input}{Input}
\SetKwInOut{Output}{Output}
\SetKwInOut{Parameter}{Parameter}

\Input {Utility functions $u, v$. Total number of periods $T$.}
\Parameter {$\eps > 0$.}

\DontPrintSemicolon

For every state $\omega_i\in \Omega$, $i=2, \ldots, |\Omega|$, apply Algorithm \ref{alg:any-pair-of-states} to $(\omega_i, \omega_1)$ with parameter $\eps$ to obtain an estimation $\hat \rho_{i1}$ of the ratio $\frac{\mu^*(\omega_i)}{\mu^*(\omega_1)}$.  Let $T_0$ be the number of periods taken in this process.  \;

Compute a prior estimation $\hat \mu \in \Delta(\Omega)$: 
$\begin{cases}
\hat\mu(\omega_1) = 1/\big(1 + \sum_{i=2}^{|\Omega|} \hat \rho_{i1}\big); \\
\hat \mu(\omega_i) = \hat \rho_{i1} \hat \mu(\omega_1), \text{ for } i=2, \ldots, |\Omega|.
\end{cases} 
$  \; 


According to Claim~\ref{claim:good-prior-estimation}, we have $\| \hat \mu - \mu^* \|_1 \le \frac{6|\Omega|\eps}{{p_0^3}}$.
Then, apply Lemma~\ref{lem:robust-pi} to the optimal signaling scheme for $\hat \mu$ with parameter $\frac{6|\Omega|\eps}{{p_0^3}}$ to obtain a robust signaling scheme $\tilde \pi$. 
Use $\tilde \pi$ for the remaining $T - T_0$ periods. 
\end{algorithm}

\begin{claim} \label{claim:good-prior-estimation}
The estimated prior $\hat \mu$ satisfies 
 $\| \hat \mu - \mu^* \|_1 \le \frac{6|\Omega|\eps}{p_0^3}$. 
\end{claim}
\begin{proof}
    See Appendix~\ref{proof:good-estimamtion}. 
\end{proof}

\begin{theorem}
\label{thm:prior-aware}
Assume Assumptions~\ref{ass:prior-p0}, \ref{ass:D}, \ref{ass:receiver-optimal-action}, \ref{ass:two-states-different-actions}. 
Choose $\eps = \frac{p_0^5 D}{42|\Omega| T}$.  The regret of Algorithm~\ref{alg:known-prior} is at most
\begin{align*}
    \Reg(T; \mathcal I) ~ \le ~ \frac{2|\Omega|}{p_0} \log_2 \frac{42|\Omega|T}{G^2p_0^6 D} + 2  ~ = ~ O\big( \log T \big). 
\end{align*}
\end{theorem}


\begin{proof}
According to Lemma~\ref{lem:relax-ratio}, each execution of Algorithm~\ref{alg:any-pair-of-states} takes at most $\frac{2}{p_0} \log_2 \frac{1}{G^2p_0 \eps}$ periods in expectation.  So the expected number $T_0$ of periods taken by the $|\Omega| - 1$ executions of Algorithm~\ref{alg:any-pair-of-states} in Line 1 of Algorithm \ref{alg:known-prior} is at most
\begin{align*} 
    \E[T_0] ~ \le ~ \tfrac{2 |\Omega|}{p_0} \log_2 \tfrac{1}{G^2p_0 \eps}. 
\end{align*} 
By Claim~\ref{claim:good-prior-estimation}, the estimated prior $\hat \mu$ satisfies $\| \hat \mu - \mu^* \|_1 \le \frac{6|\Omega|\eps}{p_0^3}$.  Using Lemma~\ref{lem:robust-pi} on $\hat \mu$ with parameter $\frac{6|\Omega|\eps}{p_0^3}$ (the condition $\frac{6|\Omega|\eps}{p_0^3} \le \frac{p_0^2D}{2}$ of the lemma is satisfied by our choice of $\eps$), the constructed scheme $\tilde \pi$ is persuasive for $\mu^*$ and is $$\tfrac{14(\frac{6|\Omega|\eps}{p_0^3})}{p_0^2 D} = \tfrac{84|\Omega|\eps}{p_0^5 D}$$
-approximately optimal for $\mu^*$.  Thus, the regret of Algorithm \ref{alg:known-prior} is upper bounded by
\begin{align*}
    \Reg(T; \mathcal I) & ~ \le ~ \underbrace{ \E[ T_0 \cdot 1]}_{\text{regret in the first $T_0$ periods}} ~ + ~ \underbrace{(T - \E[T_0]) \tfrac{84|\Omega|\eps}{p_0^5 D}}_{\text{regret in the remaining periods}} \\
    & ~ \le ~  \tfrac{2|\Omega|}{p_0} \log_2 \tfrac{1}{G^2p_0 \eps} ~ + ~ T \cdot \tfrac{84|\Omega|\eps}{p_0^5D}  ~~ = ~~ \tfrac{2|\Omega|}{p_0} \log_2 \tfrac{42|\Omega|T}{G^2p_0^6 D} ~ + ~ 2
\end{align*}
with the precision parameter chosen to be $\eps = \frac{p_0^3 D}{42|\Omega|T}$. 
\end{proof}

We now provide a brief intuition for Theorem \ref{thm:prior-aware}. In the exploration phase, the designer repeatedly uses certain \textit{non-persuasive} signaling schemes to test a hypothesized range about the unknown prior. The receiver’s action at each period helps the designer rule out part of the search interval for the unknown prior. That process is repeated over multiple periods, where the designer halves the remaining uncertainty in the estimated prior ratio in each round.
Starting from an initial interval of length $r^{(0)}-l^{(0)}$, it only takes about $\log_2(\frac{r^{(0)}-l^{(0)}}{\eps}) = O(\log T)$ rounds to reduce the interval to length $\varepsilon = O( \frac{1}{T})$, achieving the desired estimation accuracy. 


Once the prior is estimated with accuracy $\varepsilon = O(\frac{1}{T})$, the designer will construct a signaling scheme that is optimal for the prior estimate. She then applies robustification to that signaling scheme so that the robustified signaling scheme is persuasive and approximately optimal for the prior $\mu^*$. Over the remaining $T - O(\log T)$ periods, the designer incurs at most $ O(\frac{1}{T})$ regret per period, which adds up to only $O(1)$ regret in total. Summing up both the exploration phase and the exploitation phase, the total regret is at most 
\begin{align*}
    \underbrace{O(\log T)}_{\text{exploration}} + \underbrace{O(1)}_{\text{exploitation}} = O(\log T).
\end{align*}



We remark that our $O(\log T)$-regret algorithm is significantly better than using the empirical distribution of states to estimate the prior.
The empirical estimation algorithm requires the designer to observe the states (or use signaling schemes to reveal the states) and incurs at least $\Omega(\sqrt{T})$ regret due to sampling error. 
Our binary-search-based algorithm works efficiently by leveraging the structure of the receiver’s response to quickly narrow down the correct prior with much fewer rounds of interaction.

\subsection{$\Omega(\log T)$ Regret Lower Bound}

This subsection establishes a regret lower bound that matches the upper bound $O(\log T)$, showing that any algorithm must incur at least $\Omega(\log T)$ regret in the general case. That confirms that the upper bound in Theorem~\ref{thm:prior-aware} is \emph{tight} in its dependence on $T$, specifically of order $\log T$.


{Proving a lower bound is challenging because we need to argue that for \emph{any} possible learning algorithm, there exists an instance for which the algorithm's regret is at least $\Omega(\log T)$.  That requires constructing different instances for different algorithms. But thanks to von Neumann or Yao's minimax principle, proving the lower bound is equivalent to constructing a single distribution over instances such that all algorithms suffer at least $\Omega(\log T)$ regret in expectation over the distribution, which significantly simplifies the argument. In particular, we construct a distribution over information design instances with binary states and three actions where the receiver takes a sender-prefer action only when the receiver's posterior belief falls inside some region of size $\eps = O(\frac{1}{T})$ in the belief simplex. Because the designer does not know the receiver's prior, she can only try different signaling schemes to test whether the induced posterior beliefs fall inside the desired region.  By an information-theoretic argument, any algorithm that searches for that region must take at least $\Omega(\log(\frac{1}{\eps})) = \Omega(\log T)$ steps, incurring a regret of at least $\Omega(\log T) $ for the designer. 
}

\begin{theorem}
\label{thm:general-case-lower-bound}
For any learning algorithm of the designer, there exists an instance $\mathcal I$
such that the designer's regret $\Reg(T; \mathcal I)$ is at least $\Omega( \log T)$. 
\end{theorem}

We conclude this section by providing a remark on our first main result for the general setting: we obtain a tight regret bound of order $\Theta(\log T)$, which shows that the regret upper bound of our Algorithm \ref{alg:known-prior} matches the problem’s regret lower bound in order of magnitude. 
The lower bound reflects the minimal exploration cost that the designer must pay for the unknown prior, while the upper bound shows that our algorithm meets this fundamental limit. In other words, the learning speed we achieve is not only fast, but provably 
optimal up to constant factors — no other algorithm can improve the order.

\section{Faster Learning to Persuade Receivers: Binary-Action Case}
\label{sec:binary-action}

Many real-world information design problems naturally involve only two actions. For example, on the online dating platform, the male user chooses to start a conversation with the recommended female or not; an investor on an online financial platform chooses to invest in the recommended asset or not; an online video platform user chooses whether to click the recommended video.
This section thus focuses on information design with unknown prior in such binary-action scenarios.

In the binary-action case, we show that the previous $O(\log T)$-regret result for the general case can be significantly improved: there exists a learning algorithm achieving $O(\log \log T)$ regret.
We also provide a regret lower bound of $\Omega(\log \log T)$.
The matching upper and lower bounds give a \textit{tight} characterization of the performance of the learning algorithm, exactly pinning down the designer’s best-possible learning rate.

\paragraph*{Model} We now introduce the model in more detail.  The action space is now $A = \{0, 1\}$. To simplify notation, we denote the set of states by integers $\Omega = \{1, 2, \ldots, |\Omega|\}$. The designer prefers the receiver to take action $1$ over action $0$ in every state: $u(1, \omega) > u(0, \omega), \forall \omega \in \Omega$.\footnote{That is more general than the \emph{state-independent utility model} commonly studied in previous works (e.g., \cite{feng_online_2022}), where $u(1, \omega) = 1, u(0, \omega) = 0, \forall \omega \in \Omega$.}
For example, the dating platform always prefers that the male user to start a conversation with a recommended female. We assume that the receiver prefers action $0$ at the prior $\prior$, i.e.,
\begin{equation} \label{eq:assumption:prior-0}
    \sum_{\omega \in \Omega} \prior(\omega) \big( v(0, \omega) - v(1, \omega) \big) > 0. 
\end{equation}
For example, the male user prefers not to start a conversation with anyone without information provided by the platform.
Otherwise, an optimal strategy for the designer is to just reveal no information and the receiver willingly takes action $1$. 
We provide upper and lower bounds on the regret of the information designer in Sections \ref{subsec:binary-upperbound} and \ref{subsec:binary-lowerbound}, respectively.

\subsection{$O(\log \log T)$ Regret Upper Bound}
\label{subsec:binary-upperbound}

To prove the upper bound on regret in the binary-action setting, we proceed in three steps. First, we characterize the optimal signaling scheme when the prior is known to the designer. We prove that the designer’s strategy can be indexed by a single scalar, the total probability of recommending action 1, which reduces the learning problem to a one-dimensional search problem. Second, we design a learning algorithm that identifies that scalar using only the receiver’s actions. The algorithm runs in three phases: a coarse search to find a feasible interval for that scalar, a refined search to reach a precision of $O(\frac{1}{T})$, and an exploitation phase. Finally, we analyze the regret from each phase and show that the total regret is bounded by $O(\log \log T)$.

\paragraph*{Characterization of the optimal signaling scheme.}
In the binary-action case, when the prior $\prior$ is known, the optimal signaling scheme $\pi^*$ for the designer has a special structure.  As we mentioned in Section~\ref{sec:prelim-BP}, there exists an optimal signaling scheme that is direct and persuasive with signal space $S = \{0, 1\}$ and can be characterized by a linear program as follows:
\begin{align}
    \pi^* = \argmax_{\pi} & \sum_{\omega \in \Omega} \prior(\omega) \Big( \pi(1|\omega) u(1, \omega) + \pi(0|\omega) u(0, \omega) \Big).  \label{eq:LP-optimal-binary-action} \\
    \text{ s.t.} \quad & \sum_{\omega \in \Omega} \prior(\omega) \pi(1|\omega) \big( v(1, \omega) - v(0, \omega) \big) \ge 0 && \text{(persuasive for action $1$)}  \nonumber \\
    & \sum_{\omega \in \Omega} \prior(\omega) \pi(0|\omega) \big( v(0, \omega) - v(1, \omega) \big) \ge 0 && \text{(persuasive for action $0$)} \nonumber \\
    & 0\le \pi(1|\omega) = 1 - \pi(0|\omega) \le 1, \quad \forall \omega \in \Omega. \nonumber  
\end{align}
This linear program turns out to be equivalent to a \emph{fraction knapsack} problem, whose solution can be found by a greedy algorithm.  Specifically, we first divide the set of states $\Omega = \{1, \ldots, |\Omega|\}$ into two parts: the set of states with non-positive utility difference 
\begin{equation}
    \Omega^- = \Big\{\omega\in \Omega: v(0, \omega) - v(1, \omega) \le 0 \Big\} = \big\{1, \ldots, n^- \big\} 
\end{equation}
and the set of states with positive utility difference 
\begin{equation}
    \Omega^+ = \Big\{\omega\in \Omega: v(0, \omega) - v(1, \omega) > 0 \Big\} = \big\{n^-+1, \ldots, |\Omega| \big\}.
\end{equation}
Then, we sort the states in $\Omega^+$ in the following decreasing order: 
\begin{equation}
     \frac{u(1, n^-+1) - u(0, n^-+1)}{v(0, n^-+1) - v(1, n^-+1)} ~ \ge ~ \cdots ~ \ge ~ \frac{u(1, |\Omega|) - u(0, |\Omega|)}{v(0, |\Omega|) - v(1, |\Omega|)} ~ > ~ 0. 
\end{equation}
The optimal signaling scheme $\pi^*$ sends signal $1$ with probability 1 in state $\omega \in \Omega^-$ and sends signal $1$ as frequently as possible in state $\omega \in \Omega^+$ until the persuasiveness constraints are binding. Formally, with a threshold state $\omega^\dagger$ defined by
\begin{equation}
    \omega^\dagger = \min \Big\{ \omega \in \Omega : \sum_{j=1}^\omega \prior(j) \big( v(0, j) - v(1, j) \big) > 0 \Big\},
\end{equation}
the optimal signaling scheme $\pi^*$ for prior $\mu^*$ is given by the following lemma: 
\begin{lemma}[Optimal signaling scheme in the binary-action case]\label{lem:optimal-signaling-scheme-binary-action}
The following signaling scheme $\pi^*$ is an optimal solution to linear program \eqref{eq:LP-optimal-binary-action}: 
\begin{align} \label{eq:definition-pi*}
    \begin{cases}
        \pi^*(1|\omega) = 1, & \text{ for } \omega = 1, \ldots, \omega^\dagger - 1, \\ 
        \pi^*(1|\omega^\dagger) = \frac{- \sum_{j=1}^{\omega^\dagger-1} \prior(j)(v(0, j) - v(1, j))}{\prior(\omega^\dagger)(v(0, \omega^\dagger) - v(1, \omega^\dagger))}, \\
        \pi^*(1|\omega) = 0 & \text{ for } \omega = \omega^\dagger+1, \ldots, |\Omega|.
    \end{cases}
\end{align}
($\pi^*(0|\omega) = 1 - \pi^*(1|\omega)$ for all $\omega\in\Omega$.)
\end{lemma}
\noindent The proof of this lemma is in Appendix~\ref{app:proof:optimal-signaling-scheme-binary-action}. 

Then, we introduce a way to use a parameter $M$ to identify a corresponding signaling scheme $\pi^M$. In that way, there is also a parameter, say $M^*$, for the optimal signaling scheme $\pi^*$. Thus, we can construct efficient learning algorithms to search for the parameter $M^*$ to approach $\pi^*$.

\paragraph*{A parameterized family of signaling schemes $\pi^M$.}
Given the optimal signaling scheme $\pi^*$ above, we can identify a corresponding number, \textit{the optimal persuasion strength} $M^*$, which is the total conditional probability that the signal $1$ is sent:
\begin{equation}
    M^* = \sum_{\omega=1}^{|\Omega|} \pi^*(1|\omega).  
\end{equation}
Conversely, given any \textit{persuasion strength} parameter $M \in [0, |\Omega|]$, we can construct a signaling scheme $\pi^M$ whose total conditional probability $\sum_{\omega=1}^{|\Omega|} \pi(1|\omega)$ for signal $1$ is equal to the persuasion strength $M$: 
\begin{equation}
\begin{cases}
    \pi^M(1|\omega) = 1 & \text{for}~~ \omega = 1, \ldots, \lfloor M \rfloor \\
    \pi^M(1 | \omega=\lfloor M \rfloor+1) = M - \lfloor M \rfloor \\
    \pi^M(1|\omega) = 0 &\text{for}~~\omega = \lfloor M \rfloor +2, \ldots, |\Omega|. 
\end{cases}
\end{equation}
Such a signaling scheme $\pi^M$ is unique given the fixed order of states $1, 2, \ldots, |\Omega|$. 
We call $\pi^M$ \textit{a signaling scheme parameterized by persuasion strength $M \in [0, |\Omega|]$}. 
The optimal signaling scheme is thus $\pi^* = \pi^{M^*}$ with the optimal persuasion strength $M^*$. 

\begin{lemma}[Properties of $\pi^M$]
\label{lem:properties-of-pi-M}
A signaling scheme $\pi^M$ parameterized by persuasion strength $M \in [0, |\Omega|]$
\begin{itemize}
\item is persuasive for action $0$,
\item is persuasive for action $1$ if and only if $M \le M^*$,
\item and incurs a regret, compared to the optimal scheme $\pi^*$, that is bounded by $M^* - M$ for any $M \le M^*$, i.e., $U(\prior, \pi^*) - U(\prior, \pi^M) \le M^* - M$.
\end{itemize}
\end{lemma}

\noindent The proof of Lemma \ref{lem:properties-of-pi-M} is in Appendix~\ref{app:proof-properties-of-pi-M}. 

\paragraph*{Learning algorithm.} Using the above parameterized family of signaling schemes $\{\pi^M\}$, we are then able to design a learning algorithm for the designer who does not know the prior $\prior$, thereby not knowing the persuasion strength $M^*$ of the optimal signaling scheme $\pi^*$.  The idea of the algorithm is to search for $M^*$ by experimenting with signaling schemes $\pi^M$ under different values of $M$. The supporting reason is that, by Lemma~\ref{lem:properties-of-pi-M}, we can tell whether $M^* \ge M$ or $M^* < M$ by checking whether the signaling scheme $\pi^M$ that the algorithm has experimented with is persuasive or not. 
Once we obtain a value of $M$ that is within a reasonable distance from $M^*$, in particular $M^* - \frac{1}{T} \le M \le M^*$, we will continue using the signaling scheme $\pi^M$ until the end, which incurs a regret of at most $T \cdot \frac{1}{T} = O(1)$.  The searching algorithm is designed carefully such that the total regret from the searching phase is at most $O(\log \log T)$.  

We then describe our learning algorithm in detail. First, we define a subroutine, Algorithm~\ref{alg:check-persuasive}, to test whether a given signaling scheme $\pi^M$ with persuasion strength $M$ is persuasive or not by running it for multiple periods and observing the receiver's action: 

\begin{algorithm}[H]
\caption{\texttt{CheckPers}($\pi^M$)}
\label{alg:check-persuasive}

\SetKwInOut{Input}{Input}
\SetKwInOut{Output}{Output}

\Input {direct signaling scheme $\pi^M$} 
\Output {\texttt{True}/\texttt{False}, meaning whether $\pi^M$ is persuasive or not}

\DontPrintSemicolon
\LinesNumbered
Keep using signaling scheme $\pi^M$ every period until signal $s^{(t)} = 1$ is sent. \;
\eIf {\rm the receiver's action $a^{(t)} = 1$}{
	\Return \texttt{True}
}{	\Return \texttt{False} }
\end{algorithm}

With the subroutine Algorithm~\ref{alg:check-persuasive}, Algorithm~\ref{alg:binary-action} is presented below, which consists of a searching phase including two sub-phases and an exploitation phase. 

\begin{algorithm}[H]
\caption{Persuasion while Searching for $M^*$}
\label{alg:binary-action}

\SetKwInOut{Input}{Input}
\SetKwInOut{Output}{Output}

\Input {Utility functions $u, v$.  Number of periods $T$.} 
\DontPrintSemicolon
\LinesNumbered
{\color{blue} \tcp*[l]{searching phase I - finding $\underline{M}$ such that $\underline{M} \le M^* < 2\underline{M}$} }
Initialize $\underline{M} \gets |\Omega|$ \; 
\While {\rm \texttt{CheckPers}($\pi^{\underline{M}}$) = \texttt{False}} {
    $\underline{M} \gets \underline{M} / 2$ \; 
}

{\color{blue} \tcp*[l]{searching phase II - finding an interval $[L, R] \ni M^*$ with length $R - L \le \frac{1}{T}$} }
Initialize $L \gets \underline{M}, R\gets \underline{2M}$  \;   
\While {$R - L > \frac{1}{T}$} {
    $\eps \gets \frac{(R-L)^2}{2L}$, $i \gets 1$ \; 
    \While {\rm \texttt{CheckPers}$(\pi^{M = L + i \eps})$ = \texttt{True}}{
        $i \gets i + 1$ \;
    }
    $[L, R] \gets [L + (i-1) \eps, L + i \eps]$ \; 
}

{\color{blue} \tcp*[l]{exploitation phase} }
Use $\pi^{M = L}$ in all remaining periods \;
\end{algorithm}

\begin{theorem}
\label{thm:binary-action-upper-bound}
Assume Assumption \ref{ass:prior-p0} and Inequality \eqref{eq:assumption:prior-0} but no other assumptions.
In the binary-action case of information design with unknown prior, 
Algorithm~\ref{alg:binary-action} achieves a regret of at most
\begin{equation*}
    \frac{1}{p_0}\Big(7 + 3 \log_2 \log_2 (2|\Omega|T) \Big) + 1 ~ = ~ O(\log \log T). 
\end{equation*}
\end{theorem}

\begin{proof}
See Appendix \ref{app:binary-action-upper-bound}. 
\end{proof}

We now provide an intuition for Theorem \ref{thm:binary-action-upper-bound}. The $O(\log \log T)$ regret arises from the efficiency of learning a one-dimensional persuasion strength. Each interaction reveals whether the current signaling scheme crosses the persuasion strength, allowing the designer to narrow down the feasible interval exponentially over time. To achieve the precision needed for regret at most $\frac{1}{T}$, it suffices to conduct a logarithmic number of iterations in the logarithm of $T$. This doubly-logarithmic rate of convergence, combined with the fact that regret only arises in rounds where the receiver rejects the recommendation, leads to the overall $O(\log \log T)$ bound.

\subsection{$\Omega(\log \log T)$ Regret Lower Bound}
\label{subsec:binary-lowerbound}

We then establish a regret lower bound that matches the upper bound in Section~\ref{subsec:binary-upperbound}. 
We show that, in the binary-action case, any algorithm that uses signaling schemes with two signals must incur a regret of at least $\Omega(\log \log T)$.\footnote{{The restriction to two signals is motivated by the fact that the Bayesian persuasion instance in this result has only two actions and two states, with signals usually interpreted as action recommendations or partial state revelations. Nevertheless, allowing more signals might bypass this $\Omega(\log \log T)$ negative result.  We leave it as future work to determine the tight regret bound for the binary-action case using more than two signals.} 
}
Therefore, $\Theta(\log \log T)$ is the \emph{tight} regret bound for the binary-action case, and our Algorithm \ref{alg:binary-action} is optimal up to a constant factor.




\begin{theorem}
\label{thm:2-action-lower-bound}
In the binary-action case of information design with unknown prior, 
for any learning algorithm that uses signaling schemes with 2 signals, there exists an instance $\mathcal I$ with only 2 states such that the algorithm's regret $\Reg(T; \mathcal I) \ge \Omega(\log \log T)$. 
\end{theorem}

The construction of the hard instance and the proof of Theorem \ref{thm:2-action-lower-bound} are in Appendix \ref{appendix:instance}.
{We provide some intuitions here. According to the characterization of the optimal signaling scheme in Section \ref{subsec:binary-upperbound}, the optimal signaling scheme $\pi^{M^*}$ is determined by an unknown optimal persuasion strength $M^*$. The learning algorithm must learn the value of $M^*$ by experimenting with different values of persuasion strength and observing binary action feedback from the receiver.  If the algorithm has a guess about the persuasion strength $M$ that is larger than the optimal persuasion strength $M^*$, then by Lemma \ref{lem:properties-of-pi-M} the corresponding signaling scheme $\pi^M$ is too aggressive and thus not persuasive, incurring a payoff of $0$ for the designer. On the other hand, if the algorithm has a guess about the persuasion strength $M$ that is no larger than the optimal persuasion strength $M^*$, then the corresponding signaling scheme $\pi^M$ is persuasive and the designer obtains a payoff of $M$ (with regret $M^* - M$). Such a problem has a similar structure to a \emph{dynamic pricing problem}, where a seller posts a price $M$ to test whether a buyer's private value $M^*$ is larger or smaller than $M$. A known $\Omega(\log \log T)$ regret lower bound for the dynamic pricing problem \citep{kleinberg_value_2003} then implies the lower bound for our information design problem. 
}


\section{Conclusion and Discussion}
\label{sec:discussion_conclusion}
We study the information design problem with an unknown prior. 
From a worst-case perspective, we construct learning algorithms such that the designer still approaches her approximate optimality, namely, no regret, at efficient speeds. The key reason that the designer achieves so is that our learning algorithms enable the designer to learn the unknown prior 
efficiently 
from the receivers' actions.

Our work provides a learning foundation for the information design problems challenged by robustness concerns. We offer a brief discussion of several natural extensions. 

First, although we assume that the receivers are myopic and know the true distribution, we may alternatively consider a forward-looking \textit{strategic} receiver who has his own \textit{subjective} prior belief, which may be different from the states' true distribution. In that way, the receiver may be motivated to manipulate the designer's learning about his belief by taking some suboptimal actions on purpose. \textit{Will the designer still be able to learn the strategic receiver's subjective prior belief and achieve no regret?} Our construction of the learning algorithm, namely, binary search, provides the designer with a way to elicit information about the receiver's belief from his behaviors, though additional considerations should be carefully made for the designer to avoid manipulation.

Second, while our regret is defined with the benchmark being the designer's per-period Bayesian persuasion optimality for the known prior, one may be interested in studying problems in which the equilibrium solution is different from ours. For example, confronting a strategic receiver, the designer may use a \textit{globally optimal} signaling strategy. In that case, the benchmark for the designer's regret may not be based on per-period Bayesian optimality anymore.

Finally, although we assume the states are i.i.d.~distributed, 
one may consider a general case in which the states are drawn adversarially over the learning process. In that way, a new regret notion should be defined accordingly (e.g., \cite{Camara2020MechanismsFA}). It is interesting to know whether the designer in that case can still achieve no regret, and if yes, whether the designer's learning can be fast.

\bibliographystyle{te} 
\bibliography{bib.bib}  

\appendix 

\section{Useful Lemmas}
\label{section:useful_lemmas}
\begin{lemma}[Lipschitz continuity of posterior]
\label{lem:continuity-posterior}
Let $\pi : \Omega \to \Delta(S)$ be any signaling scheme.  Let $\mu, \mu' \in \Delta(\Omega)$ be two priors.  Let $\mu_s$, $\mu'_s$ be the posterior belief induced by signal $s$ under $\pi$ and prior $\mu$, $\mu'$ respectively.  Suppose $\min_{\omega \in \Omega} \mu(\omega) \ge p_0 > 0$.  Then, $\| \mu_s - \mu'_s \|_1 \le \frac{2}{p_0} \| \mu - \mu' \|_1$. 
\end{lemma}
\begin{proof}
Let $\pi(s) = \sum_{\omega \in \Omega} \mu(\omega) \pi(s|\omega)$ and $\pi'(s) = \sum_{\omega \in \Omega} \mu'(\omega) \pi(s|\omega)$ be the probability of signal $s$ under prior $\mu$ and $\mu'$ respectively.  By the definition of $\mu_s, \mu'_s$ and by triangle inequality, 
\begin{align*}
    \| \mu_s - \mu'_s \|_1 & = \sum_{\omega \in \Omega} \big| \tfrac{\mu(\omega) \pi(s|\omega)}{\pi(s)} - \tfrac{\mu'(\omega) \pi(s|\omega)}{\pi'(s)} \big| \\
    & \le \sum_{\omega \in \Omega} \big| \tfrac{\mu(\omega) \pi(s|\omega)}{\pi(s)} - \tfrac{\mu'(\omega) \pi(s|\omega)}{\pi(s)} \big| + \sum_{\omega \in \Omega} \big| \tfrac{\mu'(\omega) \pi(s|\omega)}{\pi(s)} - \tfrac{\mu'(\omega) \pi(s|\omega)}{\pi'(s)} \big|.
\end{align*}

For the first term above, 
\begin{align*}
    \sum_{\omega \in \Omega} \big| \tfrac{\mu(\omega) \pi(s|\omega)}{\pi(s)} - \tfrac{\mu'(\omega) \pi(s|\omega)}{\pi(s)} \big| = \sum_{\omega \in \Omega} \tfrac{\pi(s|\omega)}{\pi(s)} |\mu(\omega) - \mu'(\omega)|.
\end{align*}
We note that, $\forall \omega \in \Omega$, 
\begin{align}\label{eq:pi-ratio-le-p0}
    \tfrac{\pi(s|\omega)}{\pi(s)} = \tfrac{\pi(s|\omega)}{\sum_{\omega' \in \Omega} \mu(\omega') \pi(s|\omega')} \le \tfrac{\pi(s|\omega)}{p_0 \sum_{\omega' \in \Omega}\pi(s|\omega')} \le \tfrac{1}{p_0}. 
\end{align}
Thus, 
\begin{align*}
    \sum_{\omega \in \Omega} \big| \tfrac{\mu(\omega) \pi(s|\omega)}{\pi(s)} - \tfrac{\mu'(\omega) \pi(s|\omega)}{\pi(s)} \big| \le \sum_{\omega \in \Omega} \tfrac{1}{p_0} |\mu(\omega) - \mu'(\omega)| = \tfrac{1}{p_0} \| \mu - \mu' \|_1.   
\end{align*}

For the second term, 
\begin{align*}
    \sum_{\omega \in \Omega} \big| \tfrac{\mu'(\omega) \pi(s|\omega)}{\pi(s)} - \tfrac{\mu'(\omega) \pi(s|\omega)}{\pi'(s)} \big| & = \sum_{\omega \in \Omega} \mu'(\omega) \pi(s|\omega) \big| \tfrac{\pi'(s) - \pi(s)}{\pi(s) \pi'(s)} \big| \\
    & = \sum_{\omega \in \Omega} \mu'(\omega) \pi(s|\omega) \big| \tfrac{\sum_{\omega'\in\Omega} (\mu'(\omega') - \mu(\omega')) \pi(s|\omega')}{\pi(s) \pi'(s)} \big| \\
    & \le \sum_{\omega \in \Omega} \mu'(\omega) \pi(s|\omega) \tfrac{\sum_{\omega'\in\Omega} |\mu'(\omega') - \mu(\omega')|  \cdot \max_{\omega'\in\Omega} \pi(s|\omega')}{\pi(s) \pi'(s)} \\
    & = \| \mu' - \mu \|_1 \sum_{\omega \in \Omega} \tfrac{\mu'(\omega) \pi(s|\omega)}{\pi'(s)} \tfrac{\max_{\omega'\in\Omega} \pi(s|\omega')}{\pi(s)} \\
    \text{by \eqref{eq:pi-ratio-le-p0}} ~ & \le \| \mu' - \mu \|_1 \sum_{\omega \in \Omega} \tfrac{\mu'(\omega) \pi(s|\omega)}{\pi'(s)} \tfrac{1}{p_0} \\
    & = \tfrac{1}{p_0} \| \mu' - \mu \|_1.
\end{align*}

Therefore, we obtain $ \| \mu_s - \mu'_s \|_1 \le \tfrac{2}{p_0}\| \mu' - \mu \|_1$. 
\end{proof}

\begin{lemma}[Lipschitz continuity of utility]\label{lem:utility-continuity}
Suppose that, under two priors $\mu_1$ and $\mu_2$, the receiver takes the same action $a(s)$ for any signal $s\in S$.  Then, the information designer's utility satisfies
$|U(\mu_1, \pi) - U(\mu_2, \pi) | \le \| \mu_1 - \mu_2 \|_1$. 
\end{lemma}
\begin{proof}
\begin{align*}
    & | U(\mu_1, \pi) - U(\mu_2, \pi) | \\
    & = \Big| \sum_{\omega \in \Omega} \mu_1(\omega) \sum_{s\in S} \pi(s|\omega) u(a(s), \omega) - \sum_{\omega \in \Omega} \mu_2(\omega) \sum_{s\in S} \pi(s|\omega) u(a(s), \omega) \Big| \\
    & \le \sum_{\omega \in \Omega} \big| \mu_1(\omega) - \mu_2(\omega) \big| \sum_{s\in S} \pi(s|\omega) u(a(s), \omega) \\
    & \le \sum_{\omega \in \Omega} \big| \mu_1(\omega) - \mu_2(\omega) \big| \cdot 1 = \| \mu_1 - \mu_2 \|_1. 
\end{align*}
\end{proof}

\section{Robustification of Signaling Scheme: Proof of Lemma \ref{lem:robust-pi}}
\label{app:robust-pi}

Let $\hat \mu \in \Delta(\Omega)$ be a prior satisfying $\min_{\omega \in \Omega} \hat \mu(\omega) \ge p_0 > 0$.   Let $B_1(\hat \mu, \eps) = \{\mu: \|\mu - \hat \mu \|_1 \le \eps \}$ be the set of priors with $\ell_1$ distance at most $\eps$ to $\hat \mu$.  Suppose $\eps \le \frac{p_0^2D}{2}$. Let
\begin{align*}
    \delta  = \tfrac{2\eps}{p_0D} \le p_0. 
\end{align*} 
Let $\hat \pi$ be a persuasive signaling scheme for prior $\hat \mu$.
We will convert $\hat \pi$ to a signaling scheme $\tilde \pi$ that satisfies the requirements of Lemma~\ref{lem:robust-pi} via two steps: 
(1) Convert $\hat \pi$ to a non-direct signaling scheme $\tilde \pi^\circ$ that satisfies the requirements. 
(2) Convert $\tilde \pi^\circ$ to a direct signaling scheme $\pi$ that still satisfies the requirements. 

\paragraph*{Step (1): Convert $\hat \pi$ to non-direct signaling scheme $\tilde \pi^\circ$.}
Let $P_{\hat \mu, \hat \pi}(a) = \sum_{\omega \in \Omega} \hat \mu(\omega) \hat \pi(a | \omega) $ be the unconditional probability that $\hat \pi$ sends signal $a$ under prior $\hat \mu$.  Let $\hat \mu_{a, \hat \pi} \in \Delta(\Omega)$ be the posterior belief induced by signal $a$ under signaling scheme $\hat \pi$ and prior $\hat \mu$: 
\begin{align*}
    \hat \mu_{a, \hat \pi}(\omega) = \frac{\hat \mu(\omega) \hat \pi(a|\omega)}{P_{\hat \mu, \hat \pi}(a)}, \quad \forall \omega \in \Omega. 
\end{align*}
Since $\hat \pi$ is persuasive for $\hat \mu$, action $a$ must be optimal for the receiver on posterior $\hat \mu_{a, \hat \pi}$:  
\begin{align*}
    \E_{\omega \sim \hat \mu_{a, \hat \pi}}[v(a, \omega) - v(a', \omega)] \ge 0, ~ \forall a'\ne a. 
\end{align*}
According to Assumption~\ref{ass:D}, there exists a belief $\eta_a \in \Delta(\Omega)$ for which $\E_{\omega \sim \eta_a}[v(a, \omega) - v(a', \omega)] \ge D$.  Consider the convex combination of $\hat \mu_{a, \hat \pi}$ and $\eta_a$ with coefficients $1-\delta, \delta$: 
\begin{align}\label{eq:delta-combination}
    \xi_a = (1 - \delta) \hat \mu_{a, \hat \pi} + \delta \eta_a. 
\end{align}
By the linearity of expectation, $a$ must be better than any other action $a'$ by $\delta D$ on belief $\xi_a$: 
\begin{align}\label{eq:delta-D}
& \E_{\omega \sim \xi_a}[v(a, \omega) - v(a', \omega)] \nonumber  \\
& = (1-\delta)\E_{\omega \sim \hat \mu_{a, \hat \pi}}[v(a, \omega) - v(a', \omega)] + \delta \E_{\omega \sim \eta_a}[v(a, \omega) - v(a', \omega)] ~  \ge ~ \delta D. 
\end{align}
Let $\xi = \sum_{a\in A} P_{\hat \mu, \hat \pi}(a) \xi_a \in \Delta(\Omega)$, and write $\hat \mu$ as the convex combination of $\xi$ and another belief $\chi \in \Delta(\Omega)$: 
\begin{align} \label{eq:convex-combination-2}
    \hat \mu ~ = ~ (1-y) \xi + y \chi ~ = ~ \sum_{a\in A} (1-y) P_{\hat \mu, \hat \pi}(a) \xi_a ~ + ~ y \chi.
\end{align}

\begin{lemma}[Proposition 1 of \cite{zu_learning_2021}]
\label{lem:small-y}
If $\delta \le p_0$, then there exist $\chi$ on the boundary of $\Delta(\Omega)$ and $y \le \frac{\delta}{p_0} \le 1$ that satisfy \eqref{eq:convex-combination-2}. 
\end{lemma}

Since \eqref{eq:convex-combination-2} is a convex decomposition of the prior $\hat \mu$, according to \cite{kamenica_bayesian_2011}, there exists a signaling scheme $\tilde \pi^\circ$ that induces posterior $\xi_a$ with total probability $(1-y) P_{\hat \mu, \hat \pi}(a)$, $\forall a\in A$, and the posterior that puts all probability on $\omega$ with total probability $y \chi(\omega)$, $\forall \omega \in \Omega$.  Namely, $\tilde \pi^\circ$ has signal space $S = A \cup \Omega$ and conditional probability
\begin{align*}
    \tilde \pi^\circ(s | \omega) = \begin{cases}
    \frac{(1-y) P_{\hat\mu, \hat \pi}(a) \xi_a(\omega)}{\hat \mu(\omega)} & \text{ for } s = a \in A; \\
    \frac{y\chi(\omega)}{\hat \mu(\omega)} & \text{ for } s = \omega \in \Omega; \\
    0 & \text{ otherwise.}
    \end{cases}
\end{align*}
It is not hard to verify that, under prior $\hat \mu$ and signaling scheme $\tilde \pi^\circ$, the posterior induced by signal $a \in A$ is equal to $\xi_a$, and the posterior induced by signal $\omega$ is the deterministic distribution on $\omega$. 

We show that the action recommendations from $\tilde \pi^\circ$ are persuasive under all priors in $B_1(\hat \mu, \eps)$.
\begin{claim}\label{claim:tilde-pi-persuasive}
Suppose $\delta \ge \tfrac{2\eps}{p_0D}$. 
For any prior $\mu \in B_1(\hat \mu, \eps)$, any action recommendation $a\in A$ from $\tilde \pi^\circ$ is persuasive. 
\end{claim}
\begin{proof}
Under priors $\hat \mu$ and $\mu$, the posteriors beliefs induced by signal $a$ from $\tilde \pi^\circ$ are $\xi_a$ and $\mu_{a, \tilde \pi^\circ}$, respectively. 
By the continuity of posterior (Lemma~\ref{lem:continuity-posterior}), $\| \mu_{a, \tilde \pi^\circ} - \xi_a \|_1 \le \tfrac{2}{p_0} \| \mu - \hat \mu \|_1 \le \tfrac{2\eps}{p_0}$. 
Since the receiver's utility is in $[0, 1]$, for any action $a' \ne a$, we have
\begin{align*}
    \E_{\omega \sim \mu_{a, \tilde \pi^\circ}}[v(a, \omega) - v(a', \omega)] & ~ \ge ~ \E_{\omega \sim \xi_a}[v(a, \omega) - v(a', \omega)] \big| - \| \mu_{a, \tilde \pi^\circ} - \xi_a \|_1 \\
    \text{by \eqref{eq:delta-D}} & ~ \ge ~ \delta D - \tfrac{2\eps}{p_0} ~ \ge ~ 0
\end{align*}
given $\delta \ge \tfrac{2\eps}{p_0D}$. 
\end{proof}


Then, we show that the information designer's utility under signaling scheme $\tilde \pi^\circ$ is close to her utility under $\hat \pi$:  
\begin{claim}\label{claim:tilde-pi-approximately-optimal}
Given $\delta \ge \tfrac{2\eps}{p_0D}$, 
the designer's utility $U(\hat \mu, \tilde \pi^\circ) \ge U(\hat \mu, \hat \pi) - \tfrac{3\delta}{p_0}$. 
\end{claim}
\begin{proof}
According to Claim~\ref{claim:tilde-pi-persuasive}, the action recommendations from $\tilde \pi^\circ$ are persuasive for prior $\hat \mu$, so the information designer's utility satisfies 
\begin{align*}
    U(\hat \mu, \tilde \pi^\circ) & = \sum_{\omega \in \Omega} \hat \mu(\omega) \Big( \sum_{a\in A} \tilde \pi^\circ(a | \omega) u(a, \omega) ~ + ~ \tilde \pi^\circ(\omega | \omega) u(a^*(\omega), \omega) \Big) \\
    & \ge (1-y) \sum_{\omega \in \Omega}\sum_{a\in A} P_{\hat \mu, \hat \pi}(a) \xi_a(\omega) u(a, \omega) ~ + ~ 0 \\
    & \ge (1-y) \sum_{\omega \in \Omega}\sum_{a\in A} P_{\hat \mu, \hat \pi}(a) \Big( \hat \mu_{a, \hat \pi}(\omega) u(a, \omega) - | \hat \mu_{a, \hat \pi}(\omega) - \xi_a(\omega) | \Big) \\ 
    & = (1-y) \sum_{a\in A} P_{\hat \mu, \hat \pi}(a) \Big( \sum_{\omega \in \Omega} \hat \mu_{a, \hat \pi}(\omega) u(a, \omega) -  \sum_{\omega \in \Omega} \delta | \eta_a(\omega) - \hat \mu_{a, \hat \pi}(\omega) | \Big) \\
    & \ge (1-y) \sum_{a\in A} P_{\hat \mu, \hat \pi}(a) \Big( \sum_{\omega \in \Omega} \hat \mu_{a, \hat \pi}(\omega) u(a, \omega) - 2 \delta \Big) \\
    & = (1-y) U(\hat \mu, \hat \pi) - (1-y) 2\delta \\
    & \ge U(\hat \mu, \hat \pi) - y - 2\delta 
    ~ \ge ~ U(\hat \mu, \hat \pi) - \tfrac{3\delta}{p_0}
\end{align*}
where the last line uses $y \le \frac{\delta}{p_0}$ from Lemma~\ref{lem:small-y}. 
\end{proof}

\paragraph*{Step (2): Convert $\tilde \pi^\circ$ to direct signaling scheme $\tilde \pi$.}
Then, we coalesce the signals of $\tilde \pi^\circ$ (whose signal space is $A\cup \Omega$) to obtain a direct signaling scheme $\tilde \pi$ (with signal space $A$): for each state $\omega$, let $\tilde \pi$ send signal $a$ when $\tilde \pi^\circ$ sends signal $a$ and when $\tilde \pi^\circ$ sends signal $\omega$ if the receiver-optimal action $a^*(\omega) = a$:
\begin{align}
    \tilde \pi(a | \omega) = \tilde \pi^\circ(a | \omega) + \mathbbm{1}[a = a^*(\omega)] \cdot \tilde \pi^\circ(\omega | \omega). 
\end{align}
For any prior $\mu \in B_1(\hat \mu, \eps)$, any action recommendation $a\in A$ from $\tilde \pi^\circ$ is persuasive by Claim \ref{claim:tilde-pi-persuasive}, and the receiver takes the optimal action when the state is revealed under $\tilde \pi^\circ$, so the coalesced signal $a\in A$ is persuasive for the receiver, under prior $\mu$. 
Because the receiver's behavior is essentially unchanged after coalescing the signals, the designer's utility under the two signaling schemes $\tilde \pi^\circ$ and $\tilde \pi$ are the same. So we have $U(\hat \mu, \tilde \pi) = U(\hat \mu, \tilde \pi^\circ) \ge U(\hat \mu, \hat \pi) - \frac{3\delta}{p_0}$ from Claim \ref{claim:tilde-pi-approximately-optimal}. So, $\tilde \pi$ satisfies the first two requirements of Lemma \ref{lem:robust-pi}. 

It remains to prove the last requirement of Lemma \ref{lem:robust-pi}. 
Recall that $\mu^*$ is a prior satisfying $\| \hat \mu^* - \hat \mu \|_1 \le \eps$.
Applying the first two claims of of Lemma \ref{lem:robust-pi} to the optimal signaling scheme $\hat \pi$ for prior $\hat \mu$, we obtain $\tilde \pi$ that is persuasive for $\mu^*$ and $\frac{3\delta}{p_0}$-optimal for $\hat \mu$.
Symmetrically applying the argument to the optimal signaling scheme $\pi^*$ for prior $\mu^*$, we obtain $\tilde \pi^*$ that is persuasive for $\hat \mu$ and $\frac{3\delta}{p_0}$-optimal for $\mu^*$.
So, we have the following chain of inequalities: 
\begin{align*}
    U(\mu^*, \tilde \pi) &  ~ \ge ~ U(\hat \mu, \tilde \pi) - \| \mu^* - \hat \mu \|_1 && \text{ Lipschitz continuity of $U$ (Lemma~\ref{lem:utility-continuity})} \\
    & ~ \ge ~ U(\hat \mu, \hat \pi) - \tfrac{3\delta}{p_0} - \| \mu^* - \hat \mu \|_1  && \text{ $\tfrac{3\delta}{p_0}$-optimality of $\tilde \pi$ for $\hat \mu$}\\
    & ~ \ge ~ U(\hat \mu, \tilde \pi^*) - \tfrac{3\delta}{p_0} - \| \mu^* - \hat \mu \|_1 && \text{ $\hat \pi$ is optimal for $\hat \mu$}\\
    & ~ \ge ~ U(\mu^*, \tilde \pi^*) - \tfrac{3\delta}{p_0} - 2\| \mu^* - \hat \mu \|_1 && \text{ Lipschitz continuity of $U$ (Lemma~\ref{lem:utility-continuity}) } \\
    & ~ \ge ~ U(\mu^*, \pi^*) - \tfrac{6\delta}{p_0} - 2\| \mu^* - \hat \mu \|_1 && \text{ $\tfrac{3\delta}{p_0}$-optimality of $\tilde \pi^*$ for $\mu^*$} \\
    & ~ \ge ~ U(\mu^*, \pi^*) - \tfrac{6\delta}{p_0} - 2\eps. 
\end{align*}
Given $\delta = \frac{2\eps}{p_0 D}$, we conclude that $\tilde \pi$ is $\frac{14\eps}{p_0^2D}$-optimal for $\mu^*$.

\section{Proofs in Section~\ref{sec:prior-aware}}
\subsection{Proof of Lemma~\ref{lem:ratio}}
\label{proof:ratio}
First, we note that the following two claims always hold during the entire execution of Algorithm~\ref{alg:binary-search}: 
\begin{itemize}
    \item If the sender uses a signaling scheme $\pi$ that satisfies $\frac{\pi(s_0|\omega_2)}{\pi(s_0|\omega_1)} = \ell^{(k)}$, then the receiver will take action $a_1$ when signal $s_0$ is sent. 
    \item If the sender uses a signaling scheme $\pi$ that satisfies $\frac{\pi(s_0|\omega_2)}{\pi(s_0|\omega_1)} = r^{(k)}$, then the receiver will take action $\tilde a \ne a_1$ when signal $s_0$ is sent. 
\end{itemize}
These two claims hold automatically for $k \ge 1$ according to the definition of $\pi^{(k)}$, $\tilde a$, $\ell^{(k+1)}$, and $r^{(r+1)}$. 
So we only need to prove the two claims for $k = 0$.  When $k=0$, if the sender uses a signaling scheme $\pi$ that satisfies $\frac{\pi(s_0 | \omega_2)}{\pi(s_0|\omega_1)} = \ell^{(0)} = 0$, then whenever signal $s_0$ is sent, the sender will believe that the state is $\omega_2$ with probability $0$ and is $\omega_1$ with probability $1$.  Since $a_1 = \argmax_{a \in A} v(a, \omega_1)$, the receiver will take $a_1$.  If the sender uses a signaling scheme $\pi$ that satisfies $\frac{\pi(s_0 | \omega_2)}{\pi(s_0|\omega_1)} = r^{(0)} = \frac{1}{G p_0}$, then whenever signal $s_0$ is sent, the difference between the receiver's utility of taking action $\tilde a = \argmax_{a\in A} v(a, \omega_2)$ and any action $a \ne \tilde a$ is  
\begin{align*}
    & \sum_{\omega \in \Omega} \mu^*(\omega) \pi(s_0 | \omega) \big[ v(\tilde a, \omega) - v(a, \omega) \big] \\
    & = \underbrace{\mu^*(\omega_1)}_{\le 1} \pi(s_0 | \omega_1) \big[ \underbrace{v(\tilde a, \omega_1) - v(a, \omega_1)}_{\ge -1} \big] + \underbrace{\mu^*(\omega_2)}_{\ge p_0} \pi(s_0 | \omega_2) \big[ \underbrace{v(\tilde a, \omega_2) - v(a, \omega_2)}_{> G} \big]\\
    & > - \pi(s_0 | \omega_1) + p_0 G \cdot \pi(s_0 | \omega_2) ~ = ~ 0.
\end{align*}
Therefore, the receiver will take action $\tilde a$. 

After the while loop of Algorithm \ref{alg:binary-search} finishes, since the above two claims hold, if the sender uses a signaling scheme that satisfies $\frac{\pi(s_0|\omega_2)}{\pi(s_0|\omega_1)} = \ell^{(k)}$, then when $s_0$ is sent the receiver will take action $a_1$, so the utility difference between $a_1$ and $\tilde a$ is $\ge 0$: 
\begin{align*}
    & \mu^*(\omega_1) \pi(s_0 | \omega_1) \big[ \underbrace{v(a_1, \omega_1) - v(\tilde a, \omega_1)}_{> G > 0} \big] + \mu^*(\omega_2) \pi(s_0 | \omega_2) \big[ v(a_1, \omega_2) - v(\tilde a, \omega_2) \big] \ge 0 \\
    & \implies \frac{\mu^*(\omega_1)}{\mu^*(\omega_2)} \ge \frac{\pi(s_0|\omega_2)}{\pi(s_0|\omega_1)} \cdot \frac{v(\tilde a, \omega_2) - v(a_1, \omega_2) }{v(a_1, \omega_1) - v(\tilde a, \omega_1) } = \ell^{(k)} \cdot \frac{v(\tilde a, \omega_2) - v(a_1, \omega_2) }{v(a_1, \omega_1) - v(\tilde a, \omega_1) } = \hat \rho. 
\end{align*}
If the sender uses a signaling scheme that satisfies $\frac{\pi(s_0|\omega_2)}{\pi(s_0|\omega_1)} = r^{(k)}$, then when $s_0$ is sent the receiver will take action $\tilde a$, so the utility difference between $a_1$ and $\tilde a$ is $\le 0$: 
\begin{align*}
    & \mu^*(\omega_1) \pi(s_0 | \omega_1) \big[ \underbrace{v(a_1, \omega_1) - v(\tilde a, \omega_1)}_{> G > 0} \big] + \mu^*(\omega_2) \pi(s_0 | \omega_2) \big[ v(a_1, \omega_2) - v(\tilde a, \omega_2) \big] \le 0 \\
    & \implies \frac{\mu^*(\omega_1)}{\mu^*(\omega_2)} \le \frac{\pi(s_0|\omega_2)}{\pi(s_0|\omega_1)} \cdot \frac{v(\tilde a, \omega_2) - v(a_1, \omega_2) }{v(a_1, \omega_1) - v(\tilde a, \omega_1) } = r^{(k)} \cdot \frac{v(\tilde a, \omega_2) - v(a_1, \omega_2) }{v(a_1, \omega_1) - v(\tilde a, \omega_1) }. 
\end{align*}
So, given $r^{(k)} - \ell^{(k)} \le \eps G$, we have 
\begin{align*}
    \frac{\mu^*(\omega_1)}{\mu^*(\omega_2)} - \hat \rho ~ \le ~  \big(r^{(k)} - \ell^{(k)}\big) \frac{v(\tilde a, \omega_2) - v(a_1, \omega_2) }{v(a_1, \omega_1) - v(\tilde a, \omega_1)} ~ \le ~ \eps G \cdot \frac{1}{G} ~ = ~ \eps.  
\end{align*}
This means that the output $\hat \rho$ of Algorithm \ref{alg:binary-search} satisfies $\hat \rho \le \frac{\mu^*(\omega_1)}{\mu^*(\omega_2)} \le \hat \rho + \eps$. 

Then, we consider the expected number of periods needed by  Algorithm \ref{alg:binary-search}. 
First, we note that, after each while loop, the difference $r^{(k)} - \ell^{(k)}$ shrinks by a half.  The algorithm terminates when $r^{(k)} - \ell^{(k)} \le \eps G$, so the total number $k$ of while loops is at most  
\begin{align}\label{eq:bound-on-k}
    k ~ \le ~ \log_2 \frac{r^{(0)} - \ell^{(0)}}{\eps G} ~ = ~ \log_2 \frac{1}{G^2p_0 \eps}.
\end{align}
Then, we consider the number of periods needed in each while loop.  By definition, this is equal to the number of periods until signal $s_0$ is sent, whose expectation depends on the signaling scheme $\pi^{(k)}$.  We construct $\pi^{(k)}$ as follows: 
\begin{itemize}
    \item If $q \le 1$, then let $\pi^{(k)}(s_0 |\omega_2) = q$, $\pi^{(k)}(s_0 |\omega_1) = 1$; 
    \item If $q > 1$, then let $\pi^{(k)}(s_0 |\omega_2) = 1$, $\pi^{(k)}(s_0 |\omega_1) = 1/q$. 
\end{itemize}
Note that this construction satisfies $\frac{\pi^{(k)}(s_0 |\omega_2)}{\pi^{(k)}(s_0 |\omega_1)} = q$, as needed in Algorithm~\ref{alg:binary-search}.  
Since one of $\pi^{(k)}(s_0 |\omega_2)$ and $\pi^{(k)}(s_0 |\omega_1)$ is 1, the probability that $s_0$ is sent in each period is at least: 
\begin{align*}
    \pi^{(k)}(s_0) ~ = ~ \mu^*(\omega_1) \pi^{(k)}(s_0 | \omega_1) + \mu^*(\omega_2) \pi^{(k)}(s_0 | \omega_2) ~ \ge ~ \min \{ \mu^*(\omega_1), \mu^*(\omega_2) \} ~ \ge ~ p_0. 
\end{align*}
So, by the property of geometric random variable, the expected number of periods until a signal $s_0$ is sent (namely, the expected number of periods in each while loop) is at most
\begin{equation*}
    \frac{1}{\pi^{(k)}(s_0)} ~ \le ~ \frac{1}{p_0}. 
\end{equation*}
So, the total number of periods does not exceed
\begin{equation*}
    k \cdot \frac{1}{p_0} ~ \le ~ \frac{1}{p_0} \log_2 \frac{1}{G^2p_0 \eps}
\end{equation*}
in expectation.

\subsection{Proof of Lemma~\ref{lem:relax-ratio}}
\label{proof:relax-ratio}
If $(\omega_i, \omega_j)$ is a pair of distinguishable states, then Lemma~\ref{lem:ratio} shows that the estimate $\hat \rho_{ij}$ returned by Algorithm \ref{alg:binary-search} satisfies $|\hat \rho_{ij} - \frac{\prior(\omega_i)}{\prior(\omega_j)}| \le \eps \le \frac{2\eps}{p_0^2}$. 

If $(\omega_i, \omega_j)$ is not a pair of distinguishable states, then by Lemma~\ref{lem:ratio}, we have the ratio estimates $\hat \rho_{ik}$ satisfying $\hat \rho_{ik} \le \frac{\prior(\omega_i)}{\prior(\omega_k)} \le \hat \rho_{ik} + \eps$ and $\hat \rho_{jk}$ satisfying $\hat \rho_{jk} \le \frac{\prior(\omega_j)}{\prior(\omega_k)} \le \hat \rho_{jk} + \eps$. So,
\begin{align*}
    \hat \rho_{ij} ~ = ~ \frac{\hat \rho_{ik}}{\hat \rho_{jk}} ~ \le ~ \frac{\frac{\prior(\omega_i)}{\prior(\omega_k)}}{\frac{\prior(\omega_j)}{\prior(\omega_k)} - \eps} & ~ = ~ \frac{\frac{\prior(\omega_i)}{\prior(\omega_j)}}{1 - \eps \frac{\prior(\omega_k)}{\prior(\omega_j)}}.  
\end{align*}
For real numbers $a \ge 0$ and $0 \le b \le \frac{1}{2}$, we have inequality
\begin{align*}
    \frac{a}{1-b} ~ = ~ \frac{a(1-b) + ab}{1-b} ~ = ~ a + \frac{ab}{1-b} ~ \le ~ a + 2ab. 
\end{align*}
Under the assumption of $\eps \le \frac{p_0}{2}$, we have $\eps \frac{\prior(\omega_k)}{\prior(\omega_j)} \le \eps \frac{1}{p_0} \le \frac{1}{2}$. So, we obtain 
\begin{align*}
    \hat \rho_{ij}  ~ \le ~\frac{\frac{\prior(\omega_i)}{\prior(\omega_j)}}{1 - \eps \frac{\prior(\omega_k)}{\prior(\omega_j)}} ~ \le ~  \frac{\prior(\omega_i)}{\prior(\omega_j)} + 2 \eps \frac{\prior(\omega_i)}{\prior(\omega_j)} \frac{\prior(\omega_k)}{\prior(\omega_j)} ~ \le ~ \frac{\prior(\omega_i)}{\prior(\omega_j)} + \frac{2 \eps}{p_0^2}. 
\end{align*}
On the other hand, 
\begin{align*}
    \hat \rho_{ij} ~ = ~ \frac{\hat \rho_{ik}}{\hat \rho_{jk}} ~ \ge ~ \frac{\frac{\prior(\omega_i)}{\prior(\omega_k)} - \eps}{\frac{\prior(\omega_j)}{\prior(\omega_k)}} & ~ = ~ \frac{\prior(\omega_i)}{\prior(\omega_j)} - \eps \frac{\prior(\omega_k)}{\prior(\omega_j)} ~ \ge ~ \frac{\prior(\omega_i)}{\prior(\omega_j)} - \frac{\eps}{p_0} ~ \ge ~ \frac{\prior(\omega_i)}{\prior(\omega_j)} - \frac{2\eps}{p_0^2}.  
\end{align*}
So, we have $|\hat \rho_{ij} - \frac{\prior(\omega_i)}{\prior(\omega_j)}| \le \frac{2\eps}{p_0^2}$. 

The running time of Algorithm~\ref{alg:any-pair-of-states} is 2 times the running time of Algorithm \ref{alg:binary-search}, which is at most $\frac{2}{p_0} \log_2 \frac{1}{G^2p_0 \eps}$ periods in expectation by Lemma~\ref{lem:ratio}.

\subsection{Proof of Claim~\ref{claim:good-prior-estimation}}
\label{proof:good-estimamtion}
The estimation $\hat \rho_{i1}$ satisfies $| \hat \rho_{i1} - \frac{\mu^*(\omega_i)}{\mu^*(\omega_1)} | \le \eps' = \frac{2\eps}{p_0^2}$ by Lemma~\ref{lem:relax-ratio}.  Since $\mu^*$ is a probability distribution, we have $1 = \sum_{\omega \in \Omega} \mu^*(\omega) = \mu^*(\omega_1) + \mu^*(\omega_1) \sum_{i=2}^{|\Omega|} \frac{\mu^*(\omega_i)}{\mu^*(\omega_1)}$, so
\begin{align*}
    \mu^*(\omega_1) = \tfrac{1}{1 + \sum_{i=2}^{|\Omega|} \frac{\mu^*(\omega_i)}{\mu^*(\omega_1)}}. 
\end{align*}
Because the function $f(x) = \frac{1}{1+x}$ is $1$-Lipschitz ($|f'(x)| = \frac{1}{(1+x)^2} \le 1$), we have 
\begin{align*}
    | \hat \mu(\omega_1) - \mu^*(\omega_1) | \le \Big| \tfrac{1}{1 + \sum_{i=2}^{|\Omega|} \hat \rho_{i1}} - \tfrac{1}{1 + \sum_{i=2}^{|\Omega|} \frac{\mu^*(\omega_i)}{\mu^*(\omega_1)}}\Big| \le \big| \sum_{i=2}^{|\Omega|} \hat \rho_{i1} - \sum_{i=2}^{|\Omega|} \tfrac{\mu^*(\omega_i)}{\mu^*(\omega_1)} \big| \le |\Omega| \eps'. 
\end{align*}
Then, consider any $i = 2, \ldots, |\Omega|$. 
\begin{align*}
    | \hat \mu(\omega_i) - \mu^*(\omega_i) | & = \big| \hat \rho_{i1} \hat \mu(\omega_1) - \tfrac{\mu^*(\omega_i)}{\mu^*(\omega_1)} \mu^*(\omega_1) \big| \\
    & \le \big| \hat \rho_{i1} - \tfrac{\mu^*(\omega_i)}{\mu^*(\omega_1)} \big| \hat \mu(\omega_1) ~ + ~ \tfrac{\mu^*(\omega_i)}{\mu^*(\omega_1)} \big| \hat \mu(\omega_1) -  \mu^*(\omega_1) \big| \\
    & \le \eps' \hat \mu(\omega_1) ~ + ~ \tfrac{\mu^*(\omega_i)}{\mu^*(\omega_1)} |\Omega|\eps'. 
\end{align*}
Therefore,
\begin{align*}
    \| \hat \mu - \mu^* \|_1 & = |\hat \mu(\omega_1) - \mu^*(\omega_1) |  +  \sum_{i=2}^{|\Omega|} | \hat \mu(\omega_i) - \mu^*(\omega_i) | \\
    & \le |\Omega|\eps' + |\Omega|\eps' + |\Omega|\eps' \sum_{i=2}^{|\Omega|} \tfrac{\mu^*(\omega)}{\mu^*(\omega_1)} \le 2|\Omega|\eps' + |\Omega| \eps' \tfrac{1}{p_0} \le \tfrac{3|\Omega|\eps'}{p_0} = \tfrac{6|\Omega|\eps}{p_0^3}.  
\end{align*}

\subsection{Proof of Theorem \ref{thm:general-case-lower-bound}}
\label{app:general-case-lower-bound}

To prove Theorem \ref{thm:general-case-lower-bound}, 
we will construct a distribution over instances for which the expected regret of any learning algorithm of the designer will be at least $\Omega(\log T)$.
Then by Yao's minimax principle, for any learning algorithm, there must exist an instance such that the algorithm's regret is at least $\Omega(\log T)$. 

Let there be 2 states $\Omega = \{0, 1\}$. We use the probability of state $1$ to denote a belief.
Let $\dis = T^{-1.5} > 0$.
Let $p_0 = 0.1$. 
Let the prior $\mu^*$ be uniformly distributed over the set $\Gamma = \{p_0, p_0 + \dis, p_0 + 2\dis, \ldots, 3p_0 \} \subseteq[0.1, 0.3]$.
Since $\Gamma$ contains $K = \frac{2p_0}{\dis}$ elements, each corresponding to an instance, we have a distribution over $K$ instances. 
The receiver has 3 actions $A = \{a, b, c\}$.
Let $\eps_v = \frac{1}{20K} = \frac{\dis}{40 p_0}$.
The receiver prefers action $a$ when his posterior belief $\mu < 1/2 - \eps_v$, action $b$ when his posterior belief $\mu > 1/2 + \eps_v$, action $c$ when his posterior belief $\mu \in [1/2-\eps_v, 1/2+\eps_v]$.\footnote{A specific construction of the receiver's utility function $v$ is the following: $v(a, 0) = 1$, $v(a, 1) = 0$; $v(b, 1) = 1, v(b, 0) = 0$; $v(c, 0) = v(c, 1) = 1/2 + \eps_v$.}
The designer likes action $c$ only: $u(c, \omega) = 1, u(a, \omega)=u(b, \omega) = 0$ for $\omega \in \{0, 1\}$. 
The designer's optimal utility in this family of instances is characterized below: 

\begin{claim}
\label{claim:lower-bound-example-optimal-payoff}
If the prior $\prior \in \Gamma$ is known, the designer's optimal utility $U^*(\prior) = \frac{2\prior}{1-2\eps_v}$. 
\end{claim}
\begin{proof}
According to the concavification approach \citep{kamenica_bayesian_2011}, the designer's optimal signaling scheme should decompose the prior $\prior$ into two posterior, one at $\mu_1 = 0$ (where the receiver takes action $a$) and the other at $\mu_2 = 1/2 - \eps_v$ (where the receiver takes action $c$). The unconditional probabilities of these two posteriors, denoted by $1-q, q$, should satisfy the Bayes plausiblity constraint $(1-q)\cdot 0 + q \cdot (1/2-\eps_v) = \prior$, which gives $q = \frac{2\prior}{1-2\eps_v}$; that is the designer's optimal expected utility. 
\end{proof}

Before proving that the expected regret of any learning algorithm of the designer is at least $\Omega(\log T)$, we present some useful lemmas.
The first lemma characterizes the receiver's behavior:
\begin{lemma} \label{lem:example-receiver-behavior}
Let the receiver have prior $\prior$. 
Given a realized signal $s \in S$ from a signaling scheme $\pi$, 
\begin{itemize}
    \item if $\frac{\pi(s|1)}{\pi(s|0)} < \frac{1/2-\eps_v}{1/2+\eps_v}\frac{1-\prior}{\prior}$, then the receiver takes action $a$; 
    \item if $\frac{\pi(s|1)}{\pi(s|0)} > \frac{1/2+\eps_v}{1/2-\eps_v}\frac{1-\prior}{\prior}$, then the receiver takes action $b$;
    \item otherwise (namely, $\frac{\pi(s|1)}{\pi(s|0)} \in \big[\frac{1/2-\eps_v}{1/2+\eps_v}\frac{1-\prior}{\prior}, \frac{1/2+\eps_v}{1/2-\eps_v}\frac{1-\prior}{\prior}\big]$), the receiver takes action $c$. 
\end{itemize}
\end{lemma}
\begin{proof}
Given $s$, the receiver's posterior belief (for the probability of state $1$) is
\begin{align*}
    \mu_{s, \pi} = \frac{\prior \pi(s|1)}{\prior \pi(s|1) + (1-\prior) \pi(s|0)} = \frac{1}{1 + \tfrac{1-\prior}{\prior} \tfrac{\pi(s|0)}{\pi(s|1)}}. 
\end{align*}
By the definition of the receiver's utility function, 
the receiver takes action $a$ if $\mu_{s, \pi} < 1/2 - \eps_v$, which gives the first conclusion of the lemma. The second and third conclusions follow similarly. 
\end{proof}

The second lemma shows that if the designer (not knowing the prior exactly) believes that the prior $\prior$ follows a ``roughly uniform'' distribution over a subset of $\Gamma$ consisting of 2 or more elements, then the designer must suffer a constant single-period regret no matter what signaling scheme she uses:
\begin{lemma}
\label{lem:regret-on-roughly-uniform-distribution}
Let $\Gamma' \subseteq \Gamma$ with $|\Gamma'| \ge 2$. Let $P'$ be a distribution over $\Gamma'$ such that $P'(\mu) \ge \frac{1}{C |\Gamma'|}, \forall \mu \in \Gamma'$, with $C \ge 1$.  If the designer believes that $\prior$ follows $P'$, then the expected single-period regret of any signaling scheme $\pi$ is at least $\E_{\prior \sim P'}[ U^*(\prior) - U(\prior, \pi)] \ge \tfrac{2p_0}{3(1-2\eps_v)C}$. 
\end{lemma}
\begin{proof}
Let $\mu_1 \in \Gamma'$ be any possible prior. We will show that if $\pi$ is $\delta$-optimal for $\mu_1$, then $\pi$ has at least $(\frac{2\mu_2}{1-2\eps_v} - \frac{\mu_2}{\mu_1} \delta)$ regret for any other prior $\mu_2 \in \Gamma' \setminus\{\mu_1\}$. 

If $\pi$ is $\delta$-optimal for $\mu_1$, then by Claim \ref{claim:lower-bound-example-optimal-payoff} and Lemma \ref{lem:example-receiver-behavior}, the designer's expected utility $U(\mu_1, \pi)$ satisfies 
\begin{align*}
    \tfrac{2\mu_1}{1-2\eps_v} - \delta ~ \le ~ U(\mu_1, \pi) & ~ = ~ \sum_{s\in S: ~ \frac{\pi(s|1)}{\pi(s|0)} \in \big[\frac{1/2-\eps_v}{1/2+\eps_v}\frac{1-\mu_1}{\mu_1}, \frac{1/2+\eps_v}{1/2-\eps_v}\frac{1-\mu_1}{\mu_1}\big]} \Pr_{\pi, \mu_1}(s) \\
    & ~ = ~ \sum_{s\in S: ~ \frac{\pi(s|1)}{\pi(s|0)} \in I_{\pi, \mu_1}} \Big( \mu_1 \pi(s|1) + (1-\mu_1) \pi(s|0) \Big)  
\end{align*}
where we denote $I_{\pi, \mu_1} = \big[\frac{1/2-\eps_v}{1/2+\eps_v}\frac{1-\mu_1}{\mu_1}, \frac{1/2+\eps_v}{1/2-\eps_v}\frac{1-\mu_1}{\mu_1}\big]$. Note that $\frac{\pi(s|1)}{\pi(s|0)} \in I_{\pi, \mu_1}$ implies
\begin{align*}
    \pi(s|0) ~ \le ~ \tfrac{1/2+\eps_v}{1/2-\eps_v} \tfrac{\mu_1}{1-\mu_1} \pi(s|1). 
\end{align*}
So, 
\begin{align*}
    \tfrac{2\mu_1}{1-2\eps_v} - \delta & ~ \le ~ \sum_{s\in S: ~ \frac{\pi(s|1)}{\pi(s|0)} \in I_{\pi, \mu_1}} \Big( \mu_1 \pi(s|1) + (1-\mu_1) \tfrac{1/2+\eps_v}{1/2-\eps_v} \tfrac{\mu_1}{1-\mu_1} \pi(s|1) \Big) \\
    & ~ = ~  \Big( \mu_1 + \tfrac{1/2+\eps_v}{1/2-\eps_v} \mu_1 \Big) \sum_{s\in S: ~ \frac{\pi(s|1)}{\pi(s|0)} \in I_{\pi, \mu_1}}  \pi(s|1) \\
    & ~ = ~ \tfrac{2 \mu_1}{1 - 2\eps_v} \sum_{s\in S: ~ \frac{\pi(s|1)}{\pi(s|0)} \in I_{\pi, \mu_1}}  \pi(s|1), 
\end{align*}
which implies
\begin{align} \label{eq:sum-pi-1}
    \sum_{s\in S: ~ \frac{\pi(s|1)}{\pi(s|0)} \in I_{\pi, \mu_1}}  \pi(s|1) ~ \ge ~ 1 - \tfrac{1-2\eps_v}{2\mu_1} \delta.
\end{align}
We then consider any other prior $\mu_2 \in \Gamma' \setminus\{\mu_1\}$.  By our construction of the set $\Gamma$, $\mu_2$ and $\mu_1$ differ by at least $|\mu_2 - \mu_1| \ge \kappa > 0$. By tedious calculation, one can verify that
\begin{align*}
    & \text{if $\mu_2 < \mu_1$, then } \tfrac{1/2+\eps_v}{1/2-\eps_v} \tfrac{1-\mu_1}{\mu_1} < \tfrac{1/2-\eps_v}{1/2+\eps_v} \tfrac{1-\mu_2}{\mu_2},  \\
    & \text{if $\mu_2 > \mu_1$, then } \tfrac{1/2+\eps_v}{1/2-\eps_v} \tfrac{1-\mu_2}{\mu_2} < \tfrac{1/2-\eps_v}{1/2+\eps_v} \tfrac{1-\mu_1}{\mu_1}. 
\end{align*}
This means that the intersection of the intervals $I_{\pi, \mu_1} \cap I_{\pi, \mu_2} = \emptyset$. So, 
\begin{align}  \label{eq:sum-pi-2}
    \sum_{s\in S: ~ \frac{\pi(s|1)}{\pi(s|0)} \in I_{\pi, \mu_2}}  \pi(s|1) ~ \le ~ 1 - \sum_{s\in S: ~ \frac{\pi(s|1)}{\pi(s|0)} \in I_{\pi, \mu_1}}  \pi(s|1)  ~ \stackrel{\text{by \eqref{eq:sum-pi-1}}}{\le} ~ \tfrac{1-2\eps_v}{2\mu_1} \delta. 
\end{align}
Via a similar derivation as $U(\mu_1, \pi)$, we derive that $U(\mu_2, \pi)$ is at most 
\begin{align*}
    U(\mu_2, \pi) ~ \le ~ \tfrac{2\mu_2}{1-2\eps_v} \sum_{s\in S: ~ \frac{\pi(s|1)}{\pi(s|0)} \in I_{\pi, \mu_2}} \pi(s|1) ~ \stackrel{\text{by \eqref{eq:sum-pi-2}}}{\le} ~ \tfrac{2\mu_2}{1-2\eps_v} \cdot \tfrac{1-2\eps_v}{2\mu_1} \delta ~ = ~ \tfrac{\mu_2}{\mu_1}\delta, 
\end{align*}
which proves that the regret of $\pi$ for prior $\mu_2$ is at least
\begin{align*}
    \Reg(\pi, \mu_2) ~ = ~ U^*(\mu_2) - U(\mu_2, \pi) ~ \ge ~ \tfrac{2\mu_2}{1-2\eps_v} - \tfrac{\mu_2}{\mu_1} \delta.  
\end{align*}

Given the above claim that ``if $\pi$ is $\delta$-optimal for prior $\mu_1$, then $\pi$ has at least $(\frac{2\mu_2}{1-2\eps_v} - \frac{\mu_2}{\mu_1} \delta)$ regret for any prior $\mu_2 \in \Gamma' \setminus\{\mu_1\}$'', there are two possible scenarios:
\begin{itemize}
\item[(1)] $\pi$ is not $\delta$-optimal for any prior in $\Gamma'$. Then, the expected regret of $\pi$ on distribution $P'$ (which is over $\Gamma'$) is at least $\delta$.
\item[(2)] $\pi$ is $\delta$-optimal for some prior $\mu_1 \in \Gamma'$. Then, $\pi$ has regret $(\frac{2\mu_2}{1-2\eps_v} - \frac{\mu_2}{\mu_1} \delta)$ for all other priors. Since $P'$ is roughly uniform, the total probability for all other priors is at least $(|\Gamma'|-1) \frac{1}{C |\Gamma'|} \ge \frac{1}{2C}$. So, the expected regret of $\pi$ is at least $\frac{1}{2C} (\frac{2\mu_2}{1-2\eps_v} - \frac{\mu_2}{\mu_1} \delta)$. 
\end{itemize}
Taking the minimum of the above two cases,  the expected regret of $\pi$ on $P'$ is at least $\min \big\{ \delta, ~ \tfrac{1}{2C} \big(\tfrac{2\mu_2}{1-2\eps_v} - \tfrac{\mu_2}{\mu_1} \delta\big) \big\}$. 
Letting $\delta = \tfrac{1}{2C} \big(\tfrac{2\mu_2}{1-2\eps_v} - \tfrac{\mu_2}{\mu_1} \delta\big)$, namely $\delta = \tfrac{2\mu_2 \mu_1}{(1-2\eps_v)(2C\mu_1 + \mu_2)}$, we obtain
\begin{align*}
 \E_{\prior \sim P'}[ \Reg(\pi, \prior) ] ~ \ge ~ \tfrac{2\mu_2 \mu_1}{(1-2\eps_v)(2C\mu_1 + \mu_2)} ~ \ge ~ \tfrac{2p_0^2}{(1-2\eps_v)(2C p_0 + p_0)} ~ \ge ~ \tfrac{2p_0}{3(1-2\eps_v)C}
\end{align*}
where the second inequality is because $\mu_1, \mu_2 \in \Gamma'$ satisfy $\mu_1, \mu_2 \ge p_0$ and the last inequality is because $C \ge 1$. 
\end{proof}

We now lower bound the expected regret of any learning algorithm on the distribution of instances constructed above. 
To simplify the proof, we assume that the algorithm can observe all the realized states $\omega^{(1)}, \ldots, \omega^{(T)}$. Note that if we can prove that the regret of this more powerful algorithm is at least $\Omega(\log T)$, then the regret of the original algorithm must also be at least $\Omega(\log T)$.  Given the realized states, the designer's belief about the unknown prior $\prior$ is no longer uniform over $\Gamma$; it should concentrate on the empirical average of the states
\begin{align*}
    \overline \mu = \frac{1}{T} \sum_{t=1}^T \omega^{(t)}. 
\end{align*}
with $O(\sqrt{\frac{1}{T}})$ estimation error.
Formally, let $\mathcal E$ be the event in which the unknown prior $\prior$ and the empirical estimate $\overline \mu$ differ at by most $\Delta = \sqrt{\frac{9p_0}{T}}$: 
\begin{align*}
    \mathcal E = \big\{\, |\mu^* - \overline{\mu}| \leq \Delta \, \big\}.
\end{align*}
By Chernoff bound, event $\mathcal E$ happens with probability at least
\begin{align}\label{eq:chernouff-bound}
    \Pr[\mathcal E] ~ \ge ~ 1 - 2\exp\Big( - \frac{T \Delta^2}{3\prior} \Big) ~ \ge ~ 1 - 2\exp\Big( - \frac{T}{9 p_0} \frac{9p_0}{T} \Big) ~ = ~ 1 - \frac{2}{e}. 
\end{align}
where we used the fact that $\prior \le 3p_0$ for $\prior \in \Gamma$. 
In the following, we assume $\mathcal E$ happens. 
We show that, given realized states $\omega^{(1)}, \ldots, \omega^{(T)}$ and event $\mathcal E$, the designer's posterior belief about $\prior$ is a roughly uniform distribution over the set
\begin{align} \label{eq:Gamma'-definition}
    \Gamma' = \Gamma \cap [\overline \mu - \Delta, \overline \mu + \Delta].
\end{align}

\begin{lemma} \label{lem:roughly-uniform}
Given $\omega^{(1)}, \ldots, \omega^{(T)}$ and $\mathcal E$, for any $\prior \in \Gamma'$, the posterior probability $P(\prior \mid \omega^{(1)}, \ldots, \omega^{(T)}, \mathcal E) \ge \frac{1}{\exp(36) |\Gamma'|}$. In addition, the size $|\Gamma'| \ge \frac{2\Delta}{\dis}$. 
\end{lemma}

\begin{proof}
For any $\mu_1, \mu_2 \in \Gamma'$, the posterior probability ratio is
\begin{align*}
    \frac{P(\mu_1 \mid \omega^{(1)}, \ldots, \omega^{(T)}, \mathcal E)}{P(\mu_2 \mid \omega^{(1)}, \ldots, \omega^{(T)}, \mathcal E)} & ~ = ~ \frac{P(\mu_1) P(\omega^{(1)}, \ldots, \omega^{(T)}, \mathcal E \mid \mu_1)}{P(\mu_2) P(\omega^{(1)}, \ldots, \omega^{(T)}, \mathcal E \mid \mu_2)} \\
    \text{(due to uniform prior)} & ~ = ~ 1\cdot \frac{P(\omega^{(1)}, \ldots, \omega^{(T)}, \mathcal E \mid \mu_1)}{P(\omega^{(1)}, \ldots, \omega^{(T)}, \mathcal E \mid \mu_2)} \\
    \text{($\mu_1, \mu_2 \in \Gamma'$ satisfy $\mathcal E$)} & ~ = ~ \frac{P(\omega^{(1)}, \ldots, \omega^{(T)} \mid \mu_1)}{P(\omega^{(1)}, \ldots, \omega^{(T)} \mid \mu_2)} \\
    & ~ = ~ \frac{\mu_1^{\sum_{t=1}^T \omega^{(t)}} (1-\mu_1)^{T - \sum_{t=1}^T \omega^{(t)}}}{\mu_2^{\sum_{t=1}^T \omega^{(t)}} (1-\mu_2)^{T - \sum_{t=1}^T \omega^{(t)}}}.
\end{align*}
Taking the logarithm (with base $e$), 
\begin{align*}
    & \log  \frac{P(\mu_1 \mid \omega^{(1)}, \ldots, \omega^{(T)}, \mathcal E)}{P(\mu_2 \mid \omega^{(1)}, \ldots, \omega^{(T)}, \mathcal E)} ~ = ~ \sum_{t=1}^T \omega^{(t)} \log \frac{\mu_1}{\mu_2} + \big(T - \sum_{t=1}^T \omega^{(t)}\big) \log \frac{1-\mu_1}{1-\mu_2} \\
    & ~ = ~ T \Big( \overline \mu \log \frac{\mu_1}{\mu_2} + \big(1 - \overline \mu \big) \log \frac{1-\mu_1}{1-\mu_2} \Big) \\
    & ~ = ~ T \Big( \overline \mu \log \big( \frac{\mu_1}{\overline \mu} \frac{\overline \mu}{\mu_2}\big) + \big(1 - \overline \mu \big) \log \big( \frac{1-\mu_1}{1-\overline \mu} \frac{1-\overline \mu}{1-\mu_2} \big) \Big) \\
    & ~ = ~ T \Big( \overline \mu \log \frac{\overline \mu}{\mu_2} + \big(1 - \overline \mu \big) \log  \frac{1-\overline \mu}{1-\mu_2} ~ - ~  \overline \mu \log \frac{\overline \mu}{\mu_1} ~ - ~ \big(1 - \overline \mu \big) \log \frac{1-\overline \mu}{1-\mu_1}  \Big) \\
    & ~ = ~ T \Big( D_{\mathrm{KL}}(\overline \mu \| \mu_2) ~ - ~ D_{\mathrm{KL}}(\overline \mu \| \mu_1) \Big). 
\end{align*}
We then upper bound the KL-divergence $D_{\mathrm{KL}}(\overline \mu \| \mu_2)$.
Using the reverse Pinsker's inequality (see, e.g., \cite{sason_reverse_2015}) and the condition $p_0 \le \mu_2 \le 3p_0 < \frac{1}{2}$, we have 
\begin{equation*}
    D_{\mathrm{KL}}(\overline \mu \| \mu_2) ~ \le ~ \frac{1}{\min\{\mu_2, 1-\mu_2\}} \big( 2|\mu_2 - \overline \mu | \big)^2 ~ = ~ \frac{4}{\mu_2} |\mu_2 - \overline \mu |^2 ~ \le ~ \frac{4}{p_0} \Delta^2.  
\end{equation*}
Therefore, we have 
\begin{align*}
    \log  \frac{P(\mu_1 \mid \omega^{(1)}, \ldots, \omega^{(T)}, \mathcal E)}{P(\mu_2 \mid \omega^{(1)}, \ldots, \omega^{(T)}, \mathcal E)} & ~ \le ~ T D_{\mathrm{KL}}(\overline \mu \| \mu_2)  ~ \le ~ \frac{4T}{p_0} \Delta^2 ~ = ~ \frac{4T}{p_0} \frac{9p_0}{T} ~ = ~ 36.
\end{align*}
This means $\frac{P(\mu_1 \mid \omega^{(1)}, \ldots, \omega^{(T)}, \mathcal E)}{P(\mu_2 \mid \omega^{(1)}, \ldots, \omega^{(T)}, \mathcal E)} \le \exp(36)$, which implies that the posterior distribution of $\prior$ is roughly uniform on $\Gamma'$ with parameter $C = \exp(36)$. 

The size $|\Gamma'| \ge \frac{2\Delta}{\dis}$ because $\Gamma'$ is a uniformly discretized grid of length $2\Delta$. 
\end{proof}

Consider any learning algorithm of the designer. We represent it as a ternary decision tree, where each decision node $x$ corresponds to a set of possible priors $\Gamma_x \subseteq \Gamma'$, and the root node corresponds to $\Gamma'$ (defined in \eqref{eq:Gamma'-definition}) which consists of $\frac{2\Delta}{\dis}$ possible priors. At each decision node $x$, the designer chooses a signaling scheme $\pi_x$, a signal will be realized, and then the designer observes the action taken by the receiver. 
Note that the receiver has 3 possible actions, and each action corresponds to a transition on the decision tree, namely, the action refines the designer's belief about the set of possible priors from $\Gamma_x$ to $\Gamma_{x'}$ where $x'$ is a child of $x$ (the specific transition rule is given by Lemma \ref{lem:example-receiver-behavior}). 

Consider any decision node $x$ with $|\Gamma_x| \ge 2$.
Because the designer's posterior belief (given $\omega^{(1)}, \ldots, \omega^{(T)}$ and $\mathcal E$) is roughly uniform over $\Gamma'$ and $\Gamma_x \subseteq \Gamma'$, the designer's posterior belief is also roughly uniform over $\Gamma_x$, with parameter $C = \exp(36)$. 
According to Lemma \ref{lem:regret-on-roughly-uniform-distribution}, the designer must suffer a $\tfrac{2p_0}{3(1-2\eps_v)\exp(36)}$ single-period regret on this decision node. 
Let $N_{\ge 2} = \{x \in \text{Nodes} : |\Gamma_x| \ge 2 \}$ be the set of such decision nodes. 
Then, the total regret of the designer during the $T$ periods, conditioning on $\mathcal E$, is at least
\begin{align*}
    \Reg(T \mid \mathcal E) & ~ \ge \sum_{x \in N_{\ge 2}} \Pr[ \text{reach node $x$} \mid \omega^{(1)}, \ldots, \omega^{(T)}, \mathcal E] \cdot \frac{2p_0}{3(1-2\eps_v)\exp(36)}. 
\end{align*}
Because the designer' posterior beliefs on $\Gamma'$ and $\Gamma_x$ are $\exp(36)$-roughly uniform,  
\begin{align*}
    \Pr[ \text{reach node $x$} \mid \omega^{(1)}, \ldots, \omega^{(T)}, \mathcal E] = \frac{\Pr[ \prior \in \Gamma_x  \mid \omega^{(1)}, \ldots, \omega^{(T)}, \mathcal E]}{\Pr[ \prior \in \Gamma'  \mid \omega^{(1)}, \ldots, \omega^{(T)}, \mathcal E]} \ge \frac{|\Gamma_x|}{|\Gamma'| \exp(36)}. 
\end{align*}
Therefore, 
\begin{align*}
    \Reg(T \mid \mathcal E) & ~ \ge \sum_{x \in N_2} \frac{|\Gamma_x|}{|\Gamma'| \exp(36)} \cdot \frac{2p_0}{3(1-2\eps_v)\exp(36)} ~ = ~  \frac{2p_0}{3(1-2\eps_v)\exp(72)} \sum_{x \in N_2} \frac{|\Gamma_x|}{|\Gamma'|}. 
\end{align*}

\begin{lemma}
For any ternary decision tree defined on $|\Gamma'|$ elements, $\sum_{x \in N_{\ge 2}} \frac{|\Gamma_x|}{|\Gamma'|} \ge \log_3 |\Gamma'| - 1$. 
\end{lemma}
\begin{proof}
Let $d_x$ denote the depth of node $x$ in the decision tree (where the root node has depth 1). Let $\mathrm{Leaves}$ be the set of nodes with $|\Gamma_x| = 1$.  We have 
\begin{align*}
    \sum_{x \in N_{\ge 2}} |\Gamma_x| & ~ = \sum_{x \in \mathrm{Leaves}} (d_x - 1) = \sum_{x \in \mathrm{Leaves}} d_x - |\Gamma'|.
\end{align*}
We then note that
\begin{align*}
    \sum_{x \in \mathrm{Leaves}} d_x & ~ = ~ - \sum_{x \in \mathrm{Leaves}} \log_3 3^{-d_x} \\
    & ~ = ~  - |\Gamma' | \sum_{x \in \mathrm{Leaves}} \frac{1}{|\Gamma'|} \log_3 3^{-d_x} \\
    & ~ \ge ~ - |\Gamma' | \log_3 \sum_{x \in \mathrm{Leaves}} \frac{3^{-d_x}}{|\Gamma'|} &  \text{(Jensen's inequality)} \\
    & ~ \ge ~ - |\Gamma' | \log_3 \frac{1}{|\Gamma'|} & \text{(Kraft–McMillan inequality)} \\
    & ~ = ~ |\Gamma' | \log_3 |\Gamma'|. 
\end{align*}
This implies $\sum_{x \in N_{\ge 2}} \frac{|\Gamma_x|}{|\Gamma'|} \ge \log_3 |\Gamma'| - 1$. 
\end{proof}

With the above lemma, we obtain
\begin{align*}
    \Reg(T \mid \mathcal E) & ~ \ge ~ \frac{2p_0}{3(1-2\eps_v)\exp(72)} \Big( \log_3 |\Gamma'| - 1\Big) \\
    & ~ = ~ \Omega\Big( p_0 \cdot \log \frac{2\Delta}{\kappa} \Big)  && \text{($|\Gamma'| \ge \frac{2\Delta}{\kappa}$ by Lemma \ref{lem:roughly-uniform})} \\
    & ~ = ~ \Omega\Big( p_0 \cdot \log_3 \big( \sqrt{\frac{9p_0}{T}}\cdot T^{1.5} \big) \Big)  && \text{(by our choice of $\Delta$ and $\kappa$)} \\
    & ~ = ~ \Omega\big( \log T \big). 
\end{align*}
Since event $\mathcal E$ happens with probability at least $1 - 2/e$ by \eqref{eq:chernouff-bound}, the total expected regret of the designer is at least $(1-2/e) \Reg(T \mid \mathcal E) \ge \Omega\big( \log T \big)$, which proves the theorem. 

\section{Proofs in Section \ref{subsec:binary-upperbound}}

\subsection{Proof of Lemma~\ref{lem:optimal-signaling-scheme-binary-action}}
\label{app:proof:optimal-signaling-scheme-binary-action}
An optimal signaling scheme is given by linear program \eqref{eq:LP-optimal-binary-action}.  Let's first simplify this linear program.  Under the constraint of $\pi(0|\pi) = 1 - \pi(1|\pi)$, the objective of the linear program can be written as
\begin{align*}
    & \argmax_{\pi} \sum_{\omega \in \Omega} \prior(\omega) \Big( \pi(1|\omega) u(1, \omega) + (1-\pi(1|\omega)) u(0, \omega) \Big) \\
     = ~ & \argmax_{\pi} \bigg\{ \sum_{\omega \in \Omega} \prior(\omega) \pi(1|\omega) \big(u(1, \omega) - u(0, \omega) \big) ~ + ~ \sum_{\omega \in \Omega} \prior(\omega) u(0, \omega) \bigg\} \\
     = ~ & \argmax_{\pi} \sum_{\omega \in \Omega} \prior(\omega) \pi(1|\omega) \big(u(1, \omega) - u(0, \omega) \big).
\end{align*}
The persuasiveness constraint for action $0$ 
\begin{align*}
    \sum_{\omega \in \Omega} \prior(\omega) \pi(0|\omega) \big( v(0, \omega) - v(1, \omega) \big) \ge 0 
\end{align*}
is equivalent to 
\begin{align*}
    & \sum_{\omega \in \Omega} \prior(\omega) \big(1-\pi(1|\omega)\big) \big( v(0, \omega) - v(1, \omega) \big) \ge 0 \\
    \iff ~ & \sum_{\omega \in \Omega} \prior(\omega) \pi(1|\omega) \big( v(1, \omega) - v(0, \omega) \big) \ge - \sum_{\omega \in \Omega} \prior(\omega) \big( v(0, \omega) - v(1, \omega) \big). 
\end{align*}
We note that the above inequality is automatically satisfied when the persuasiveness constraint for action $1$, $\sum_{\omega \in \Omega} \prior(\omega) \pi(1|\omega) \big( v(1, \omega) - v(0, \omega) \big) \ge 0$, is satisfied and Inequality \eqref{eq:assumption:prior-0}, $\sum_{\omega \in \Omega} \prior(\omega) \big( v(0, \omega) - v(1, \omega) \big) > 0$, is assumed.  So, we can drop the persuasiveness constraint for action $0$ and simplify the linear program to the following: 
\begin{align} \label{eq:fractional-knapsack}
    \pi^* ~ = ~ \argmax_{\pi} & \sum_{\omega \in \Omega} \prior(\omega) \pi(1|\omega) \big(u(1, \omega) - u(0, \omega) \big) \\
    \text{ s.t.} \quad & \sum_{\omega \in \Omega} \prior(\omega) \pi(1|\omega) \big( v(0, \omega) - v(1, \omega) \big) \le 0 && \text{(persuasive for action $1$)}  \nonumber \\
    & 0\le \pi(1|\omega) \le 1, \quad \forall \omega \in \Omega. \nonumber  
\end{align}

This is a fractional knapsack problem with item set $\Omega$ where each item $\omega \in \Omega$ has positive value $\prior(\omega)\big(u(1, \omega) - u(0, \omega)\big) > 0$ and potentially negative weight $\prior(\omega)\big(v(0, \omega) - v(1, \omega)\big)$.  We are to decide how much fraction $\pi(1|\omega) \in [0, 1]$ of each item to include, while ensuring that the total weight does not exceed $0$. 
First, it is easy to see that items with non-positive weights $\prior(\omega)\big(v(0, \omega) - v(1, \omega)\big) \le 0$ should always be included: $\pi(1|\omega) = 1$.  Then, for the remaining items (those with positive weight $\prior(\omega)\big(v(0, \omega) - v(1, \omega)\big) > 0$, it is known that the greedy strategy of including the items in the decreasing order to value-per-weight $\frac{u(1, \omega) - u(0, \omega)}{v(0, \omega) - v(1, \omega)}$ is optimal \citep{korte_combinatorial_2012}. 
Formally, we divide the item set $\Omega = \{1, \ldots, |\Omega|\}$ into two parts: the non-positive-weight part denoted by
\begin{equation}
    \Omega^- = \{\omega\in \Omega: v(0, \omega) - v(1, \omega) \le 0 \} = \{1, \ldots, n^-\}. 
\end{equation}
and positive weight part $\Omega^+ = \{\omega\in \Omega: v(0, \omega) - v(1, \omega) > 0 \} = \{n^-+1, \ldots, |\Omega|\}$, and sort the second part such that: 
\begin{equation}
     \frac{u(1, n^-+1) - u(0, n^-+1)}{v(0, n^-+1) - v(1, n^-+1)} ~ \ge ~ \cdots ~ \ge ~ \frac{u(1, |\Omega|) - u(0, |\Omega|)}{v(0, |\Omega|) - v(1, |\Omega|)} ~ > ~ 0. 
\end{equation}
Let $\omega^\dagger$ denote the first position at which we violate the total weight constraint $\sum_{\omega \in \Omega} \prior(\omega) \pi(1|\omega) \big( v(0, \omega) - v(1, \omega) \big) \le 0$ when greedily including items in the order of $1, 2, ..., n^-, n^-+1, \ldots, |\Omega|$: 
\begin{align*}
    \omega^\dagger ~ = ~ \min\Big\{ \omega \in \Omega : \sum_{j=1}^\omega \prior(j) \big( v(0, j) - v(1, j) \big) > 0\Big\}. 
\end{align*}
Then, define
\begin{align}
    \begin{cases}
        \pi^*(1|\omega) = 1, & \text{ for } \omega = 1, \ldots, \omega^\dagger - 1 \\ 
        \pi^*(1|\omega^\dagger) = \frac{- \sum_{j=1}^{\omega^\dagger-1} \prior(j)(v(0, j) - v(1, j))}{\prior(\omega^\dagger)(v(0, \omega^\dagger) - v(1, \omega^\dagger))} \\
        \pi^*(1|\omega) = 0 & \text{ for } \omega = \omega^\dagger+1, \ldots, |\Omega|
    \end{cases}
\end{align}
with $\pi^*(0|\omega) = 1 - \pi^*(1|\omega)$.  The optimality of $\pi^*$ follows from the optimality of greedy value-per-weight algorithm for fractional knapsack problems (see, e.g., \cite{korte_combinatorial_2012}). 

\subsection{Proof of Lemma~\ref{lem:properties-of-pi-M}}
\label{app:proof-properties-of-pi-M}

\paragraph*{(1) Prove that $\pi^M$ is persuasive for action $0$:}  We first note that, when $M = 0$, we have $\pi^M(1|\omega)=0$ and $\pi^M(0|\omega)=1$ for all $\omega \in \Omega$, in which case
\begin{align*}
   \sum_{\omega \in \Omega} \prior(\omega) \pi^M(0|\omega) \big( v(0, \omega) - v(1, \omega) \big) ~ = ~  \sum_{\omega \in \Omega} \prior(\omega) \big( v(0, \omega) - v(1, \omega) \big) ~ > ~ 0
\end{align*}
by Assumption~\ref{eq:assumption:prior-0}, so $\pi^M$ is persuasive for action $0$.  Then, if we increase $M$ from $0$, $\pi^M(1|\omega)$ will increase from $0$ to $1$ from state $\omega = 1, 2, \ldots, $ to state $|\Omega|$, meaning that $\pi^M(0|\omega)$ will decrease.  The differences $v(0, \omega) - v(1, \omega)$ of the first $n^-$ states are non-positive by our ordering of states, so the decrease of $\pi^M(0|\omega)$ will weakly increase the summation $\sum_{\omega \in \Omega} \prior(\omega) \pi^M(0|\omega) \big( v(0, \omega) - v(1, \omega) \big)$, so this summation continues to be $\ge 0$.  Once $M$ exceeds $n^-$, all the states $\omega = 1, \ldots, n^-$ will have $\pi^M(0|\omega) = 0$, so 
\begin{align*}
   \sum_{\omega \in \Omega} \prior(\omega) \pi^M(0|\omega) \big( v(0, \omega) - v(1, \omega) \big) ~ = ~  \sum_{\omega = n^-+1}^{|\Omega|} \prior(\omega) \pi^M(0|\omega) \big( v(0, \omega) - v(1, \omega) \big)  ~ \ge ~ 0
\end{align*}
because $v(0, \omega) - v(1, \omega) > 0$ for $\omega = n^-+1, \ldots, |\Omega|$ by our ordering of states.  So the signaling scheme $\pi^M$ is persuasive for action $0$. 

\paragraph*{(2) Prove that $\pi^M$ is persuasive for action $1$ when $M \le M^*$:}  We first note that the optimal signaling scheme $\pi^* = \pi^{M^*}$ makes the receiver indifferent between actions $1$ and $0$ when signal $1$ is sent: 
\begin{equation*}
\sum_{\omega \in \Omega} \prior(\omega) \pi^*(1|\omega) \big( v(1, \omega) - v(0, \omega) \big) ~ = ~ 0. 
\end{equation*}
One can verify this using the definition of $\pi^*$ (Equation \eqref{eq:definition-pi*}).
Moreover, we know that the threshold state $\omega^\dagger = \lfloor M^* \rfloor +1$ must be larger than $n^-$ because all the states $\omega = 1, \ldots, n^-$ satisfy $v(1, \omega) - v(0, \omega) \ge 0$ by our ordering of states.
 
If we decrease $M$ from $M^*$, the probability $\pi^M(1|\omega)$ of some states $\omega > n^-$ will start to decrease, which will increase the quantity $\sum_{\omega \in \Omega} \prior(\omega) \pi^M(1|\omega) \big( v(1, \omega) - v(0, \omega) \big)$ because $v(1, \omega) - v(0, \omega) < 0$ for $\omega > n^-$, so the signaling scheme $\pi^M$ continues to be persuasive for action $1$.  Once $M$ is decreased to the extend that $\lfloor M \rfloor + 1 \le n^-$, we have
\begin{equation*}
\sum_{\omega \in \Omega} \prior(\omega) \pi^M(1|\omega) \big( v(1, \omega) - v(0, \omega) \big) ~ = ~ \sum_{\omega=1}^{\lfloor M \rfloor + 1} \prior(\omega) \pi^M(1|\omega) \big( v(1, \omega) - v(0, \omega) \big) ~ \ge ~ 0 
\end{equation*}
because $v(1, \omega) - v(0, \omega) \ge 0$ for $\omega \le n^-$. 

If we increase $M$ from $M^*$, the probability $\pi^M(1|\omega)$ of some states $\omega > n^-$ will start to increase, which will strictly decrease the quantity $\sum_{\omega \in \Omega} \prior(\omega) \pi^M(1|\omega) \big( v(1, \omega) - v(0, \omega) \big)$ because $v(1, \omega) - v(0, \omega) < 0$ for $\omega > n^-$, so the signaling scheme $\pi^M$ will become not persuasive for action $1$.  

The above analysis proves that $\pi^M$ is persuasive for action $1$ if and only if $M \le M^*$. 

\paragraph*{(3) Prove that $U(\prior, \pi^*) - U(\prior, \pi^M) \le M^* - M$ when $M \le M^*$:} 
When $M \le M^*$, by the second property we know that $\pi^M$ is persuasive for action $1$ (and by the first property, $\pi^M$ is persuasive for action $0$), so the expected utility of the information designer using $\pi^M$ is equal to 
\begin{align*}
	U(\prior, \pi^M) & ~ = ~ \sum_{\omega \in \Omega} \prior(\omega) \Big( \pi^M(1|\omega) u(1, \omega) + \pi^M(0|\omega) u(0, \omega) \Big)  \\
	& ~ = ~ \sum_{\omega \in \Omega} \prior(\omega) \Big( \pi^M(1|\omega) u(1, \omega) + \big(1 - \pi^M(1|\omega) \big) u(0, \omega) \Big) \\
	&  ~ = ~ \sum_{\omega \in \Omega} \prior(\omega) \pi^M(1|\omega) \Big( u(1, \omega) - u(0, \omega) \Big) ~ + ~  \sum_{\omega \in \Omega} \prior(\omega) u(0, \omega).  
\end{align*}
Taking the difference with the expected utility of the optimal signaling scheme $\pi^* = \pi^{M^*}$: 
\begin{align}
    U(\prior, \pi^*) - U(\prior, \pi^M) &  ~ = ~ \sum_{\omega \in \Omega} \prior(\omega) \Big( \pi^*(1|\omega) - \pi^M (1|\omega) \Big) \Big( u(1, \omega) - u(0, \omega) \Big) \nonumber\\
     \text{because $u(a, \omega) \in [0, 1]$} & ~ \le ~  \sum_{\omega \in \Omega} \prior(\omega) \Big( \pi^*(1|\omega) - \pi^M (1|\omega) \Big).  \label{eq:U-difference}
\end{align}
With $M \le M^*$, we have $\pi^M(1|\omega) \le \pi^{M^*}(1|\omega)$.  And because $\prior(\omega) \le 1$, we have $\prior(\omega) \big( \pi^*(1|\omega) - \pi^M (1|\omega) \big) \le \pi^*(1|\omega) - \pi^M (1|\omega)$. Thus,   
\begin{align*}
    U(\prior, \pi^*) - U(\prior, \pi^M) &  ~ \le ~ \sum_{\omega \in \Omega} \Big( \pi^*(1|\omega) - \pi^M (1|\omega) \Big) \\
    & ~ = ~ \sum_{\omega \in \Omega} \pi^*(1|\omega) - \sum_{\omega \in \Omega} \pi^M (1|\omega) ~ = ~ M^* - M. 
\end{align*}




\subsection{Proof of Theorem \ref{thm:binary-action-upper-bound}}
\label{app:binary-action-upper-bound}
The (expected) regret of Algorithm~\ref{alg:check-persuasive} is defined to be $\E[ X U^* - \sum_{t=1}^X u(a^{(t)}, \omega^{(t)})]$, where $X$ is the total number of periods used by the algorithm. 
Lemma \ref{lem:check-persuasive} gives upper bounds on the regret of Algorithm~\ref{alg:check-persuasive}.

\begin{lemma}
\label{lem:check-persuasive}
Algorithm~\ref{alg:check-persuasive} satisfies the following: 
\begin{itemize}
    \item if $\pi^M$ is persuasive, i.e., $M \le M^*$, then Algorithm~\ref{alg:check-persuasive} returns {\rm \texttt{True}} and its expected regret is at most $\frac{M^* - M}{p_0 M}$; 
    \item if $\pi^M$ is not persuasive, i.e., $M > M^*$, then Algorithm~\ref{alg:check-persuasive} returns {\rm \texttt{False}} and its expected regret  is at most $\frac{M^*}{p_0 M}$.
\end{itemize}
\end{lemma}
The proof of Lemma \ref{lem:check-persuasive} is in Appendix~\ref{proof:check-persuasive}.

\begin{claim}\label{claim:searching-phase-1}
The regret from {\rm \texttt{searching phase I}} is at most $\frac{4}{p_0}$. 
\end{claim}
\begin{proof}
Let $K = \lceil \log_2 \frac{|\Omega|}{M^*} \rceil$, so $\frac{|\Omega|}{2^K} \le M^* < \frac{|\Omega|}{2^{K-1}}$. In the $k$-th iteration of the while loop ($k\ge 0$), $\underline{M}$ is equal to $\frac{|\Omega|}{2^k}$.  Thus, when $k< K$, we have $\underline{M} = \frac{|\Omega|}{2^k} > M^*$ and $\pi^{\underline{M}}$ is not persuasive by Lemma~\ref{lem:properties-of-pi-M}. Therefore, \rm \texttt{CheckPers}($\pi^{\underline{M}}$) = \texttt{False} and the while loop will continue.  The while loop ends at iteration $k = K$ where \rm \texttt{CheckPers}($\pi^{\underline{M}}$) = \texttt{True}. So, by Lemma~\ref{lem:check-persuasive}, the total regret from this phase is at most 
\begin{align*}
    \sum_{k=0}^K \frac{M^*}{p_0 \underline{M}} ~ \le ~ 
    \sum_{k=0}^K \frac{\frac{|\Omega|}{2^{K-1}}}{p_0 \frac{|\Omega|}{2^k}} ~ = ~ \frac{1}{p_0} \sum_{k=0}^K \frac{1}{2^{K-k-1}} ~ \le ~ \frac{4}{p_0}. 
\end{align*}
\end{proof}

\begin{claim}\label{claim:searching-phase-2}
The regret from {\rm \texttt{searching phase II}} is at most $\frac{3}{p_0}\big(1 +  \log_2 \log_2 (2|\Omega|T)\big)$.
\end{claim}
\begin{proof}
For the $k$-th ($k\ge 0$) iteration of the outer while loop (Lines 6 to 12), we write $L = L^{(k)}, R = R^{(k)}$, and $\eps^{(k)} = \frac{(R^{(k)} - L^{(k)})^2}{2L^{(k)}}$. The number of iterations of the inner while loop (Lines 8 to 10) is at most $\frac{R^{(k)} - L^{(k)}}{\eps^{(k)}}$ (because $\texttt{CheckPers}(\pi^{R^{(k)}}) = \texttt{False}$ given $R^{(k)} > M^*$).  By Lemma~\ref{lem:check-persuasive}, the regret of each inner while loop iteration is at most, in the case of \texttt{CheckPers}$(\pi^{M = L+i\eps}) = \texttt{True}$, 
\begin{equation*}
    \frac{M^* - M}{p_0 M}~ \le ~ \frac{R^{(k)} - L^{(k)}}{p_0 L^{(k)}}.  
\end{equation*}
and in the case of \texttt{CheckPers}$(\pi^{M = L+i\eps}) = \texttt{False}$ (when $M > M^*$), 
\begin{equation*}
    \frac{M^*}{p_0M} ~ < ~ \frac{1}{p_0}. 
\end{equation*}
Thus, the total regret from the inner while loop (Lines 8 to 10) is at most 
\begin{align*}
    \frac{R^{(k)} - L^{(k)}}{\eps^{(k)}} \cdot  \frac{R^{(k)} - L^{(k)}}{p_0 L^{(k)}} + \frac{1}{p_0} ~ = ~ \frac{(R^{(k)} - L^{(k)})^2}{\frac{(R^{(k)}-L^{(k)})^2}{2L^{(k)}}p_0 L^{(k)}} + \frac{1}{p_0} ~ = ~ \frac{3}{p_0}. 
\end{align*}
Then, what's the total number of iterations of the outer while loop (Lines 6 to 12)?  We note that the length of interval $[L^{(k+1)}, R^{(k+1)}]$ is equal to $\eps^{(k)}$ and satisfies the following recursion: 
\begin{align*}
    \frac{R^{(k+1)} - L^{(k+1)}}{2L^{(k+1)}} = \frac{\eps^{(k)}}{2L^{(k+1)}} \le \frac{\eps^{(k)}}{2L^{(k)}} & = \Big(\frac{R^{(k)}-L^{(k)}}{2L^{(k)}}\Big)^2 \\
    & \le \Big(\frac{R^{(0)}-L^{(0)}}{2L^{(0)}}\Big)^{2^{k+1}} \\
    & = \Big(\frac{2L^{(0)}-L^{(0)}}{2L^{(0)}}\Big)^{2^{k+1}} \\
    & = \Big(\frac{1}{2} \Big)^{2^{k+1}}. 
\end{align*}
Therefore, we have $R^{(k)} - L^{(k)} \le \frac{2L^{(k)}}{2^{2^k}} \le \frac{2|\Omega|}{2^{2^k}}$.  Since the outer while loop ends when $R^{(k)} - L^{(k)} \le \frac{1}{T}$, the number of iterations of the outer while loop is no more than $1 + \log_2 \log_2 (2|\Omega|T)$, so the total regret is at most $\big(1 + \log_2 \log_2 (2|\Omega|T) \big) \cdot \frac{3}{p_0}$.
\end{proof}

\begin{claim}\label{claim:exploration-phase}
    The regret from the {\rm \texttt{exploitation phase}} is at most $1$. 
\end{claim}
\begin{proof}
    The $L$ and $R$ after \texttt{search phase II} satisfy $L \le M^* < R$ and $R - L\le \frac{1}{T}$. By Lemma~\ref{lem:properties-of-pi-M}, the regret of $\pi^{L}$ in each period is at most $M^* - L \le \frac{1}{T}$. So, the total regret is at most $T \cdot \frac{1}{T} = 1$. 
\end{proof}

Summing up the regrets from the three phases in Claims~\ref{claim:searching-phase-1}, \ref{claim:searching-phase-2}, \ref{claim:exploration-phase}, we obtain the regret bound in Theorem \ref{thm:binary-action-upper-bound}.

\subsection{Proof of Lemma~\ref{lem:check-persuasive}}
\label{proof:check-persuasive}
By Lemma~\ref{lem:properties-of-pi-M}, the signaling scheme $\pi^M$ is always persuasive for action $0$.  So, whether $\pi^M$ is persuasive is equivalent to whether $\pi^M$ is persuasive for action $1$, which is verified by Algorithm \ref{alg:check-persuasive} $\texttt{CheckPers}(\pi^M)$ by design.  So, when the algorithm returns \texttt{True} the signaling scheme $\pi^M$ must be persuasive, and otherwise not persuasive.  Also by Lemma~\ref{lem:properties-of-pi-M}, we know that whether $\pi^M$ is persuasive is equivalent to whether $M \le M^*$. 

We then consider the expected regret of Algorithm \ref{alg:check-persuasive}. 
Let random variable $X$ be the total number of periods needed by Algorithm \ref{alg:check-persuasive}, which is equal to the number of periods until signal $1$ is realized.  The probability that signal $1$ is realized at each period is equal to: 
\begin{align*}
    p^M(s=1) ~ = ~ \sum_{\omega \in \Omega} \prior(\omega) \pi^M(1|\omega). 
\end{align*}
So, the expectation of $X$ is (according to the property of geometric random variable)
\begin{align*}
    \E[X] ~ = ~ \frac{1}{p^M(s=1)} ~ = ~ \frac{1}{\sum_{\omega \in \Omega} \prior(\omega) \pi^M(1|\omega)} ~ \le ~ \frac{1}{p_0 \sum_{\omega \in \Omega} \pi^M(1|\pi)} ~ = ~ \frac{1}{p_0M}. 
\end{align*}

When $\pi^M$ is persuasive (i.e., $M \le M^*$), Lemma \ref{lem:properties-of-pi-M} shows that the regret due to using signaling scheme $\pi^M$ in each period is at most 
\begin{equation*}
    U^* - U(\prior, \pi^M) ~ \le ~ M^* - M. 
\end{equation*}
So, the total expected regret of Algorithm \ref{alg:check-persuasive} is at most 
\begin{align*}
    \E[X (U^* - U(\prior, \pi^M))] ~ \le ~ \E[X] \cdot (M^* - M) ~ \le ~ \frac{M^* - M}{p_0M}. 
\end{align*}

When $\pi^M$ is not persuasive (i.e. $M > M^*$), the receiver always takes action $0$ regardless of the signal, so the expected utility of $\pi^M$ in each period is 
\begin{align*}
    U(\prior, \pi^M) = \sum_{\omega \in \Omega} \prior(\omega) u(0, \omega)
\end{align*}
and the regret per period is
\begin{align*}
    U^* - U(\prior, \pi^M) & = \sum_{\omega \in \Omega} \prior(\omega) \Big(\pi^*(1|\omega) u(1, \omega) + \pi^*(0|\omega) u(0, \omega) \Big) - \sum_{\omega \in \Omega} \prior(\omega) u(0, \omega) \\
    & = \sum_{\omega \in \Omega} \prior(\omega) \pi^*(1|\omega) \Big( u(1, \omega) - u(0, \omega) \Big) \\
    & \le \sum_{\omega \in \Omega} \prior(\omega) \pi^*(1|\omega) \quad\quad \text{ because $u(a, \omega)\in[0, 1]$} \\
    & \le \sum_{\omega \in \Omega} \pi^*(1|\omega) \\ 
    & = M^*. 
\end{align*}
So, the total expected regret is at most 
\begin{align*}
    \E[X (U^* - U(\prior, \pi^M))] ~ \le ~ \E[X] \cdot M^* ~ \le ~ \frac{M^*}{p_0 M}. 
\end{align*}

\section{Proof of Theorem \ref{thm:2-action-lower-bound}}

\label{appendix:instance} 

\paragraph*{Dynamic pricing problem.}
Our proof of the lower bound is inspired by \cite{feng_online_2022}\footnote{\cite{feng_online_2022} study online Bayesian persuasion with unknown utility and known prior.} and uses a reduction from the following \emph{single-item dynamic pricing problem} \citep{kleinberg_value_2003}. There are $T$ buyers and a seller with unlimited copies of a single item.  All buyers have the same valuation $v^* \in [0, 1]$ for the item, which is unknown to the seller.  At each period $t \in [T]$, a buyer arrives and the seller posts a price $p_t$ to the buyer, who buys the item and pays $p_t$ if and only if $v^* \ge p_t$.  The regret of a dynamic pricing algorithm $\Alg^\textrm{pricing}$ is
\begin{equation}
    \Reg(\Alg^\textrm{pricing}) ~ = ~ T\cdot v^* - \E_{p_1, \ldots, p_T}\Big[\sum_{t=1}^T p_t \cdot \mathbbm{1}[p_t \le v^*]\Big]. 
\end{equation}
\cite{kleinberg_value_2003} show that no algorithm can achieve a better regret than $\Omega(\log \log T)$. 
\begin{theorem}[\cite{kleinberg_value_2003}]
\label{thm:dynamic-pricing-lower-bound}
For any (possibly randomized) dynamic pricing algorithm $\Alg^\text{\rm pricing}$, with the valuation $v^*$ randomly sampled the uniform distribution on $[0, 1]$, the expected regret of the algorithm is at least $\Omega(\log \log T)$.  
\end{theorem}

\paragraph*{Constructing Bayesian persuasion instances.}
We construct a family of Bayesian persuasion instances such that, if the designer achieves a regret better than $\Omega(\log \log T)$ in the Bayesian persuasion problem, then there exists an algorithm that achieves a regret better than $\Omega(\log \log T)$ in the dynamic pricing problem, which contradicts Theorem~\ref{thm:dynamic-pricing-lower-bound} and proves Theorem~\ref{thm:2-action-lower-bound}. 
For each dynamic pricing problem instance, parameterized by $(T, v^*)$, we construct a corresponding Bayesian persuasion instance as follows. There are two states $\Omega = \{0, 1\}$ and two actions $A = \{0, 1\}$. 
Let $\eps = \frac{1}{T^3}$.
The prior $\prior$ is 
\begin{equation}
    \prior(\omega=0) = \frac{1}{1 + \eps v^*}, \quad \prior(\omega=1) = \frac{\eps v^*}{1+\eps v^*}. 
\end{equation}
The utility of the receiver is: 
\begin{equation}
    \text{for state $0$:} \begin{cases}
        v(a=0, \omega=0) = 0 \\
        v(a=1, \omega=0) = -\eps
    \end{cases} \quad \quad \text{for state $1$:} \begin{cases}
        v(a=0, \omega=1) = 0 \\
        v(a=1, \omega=1) = 1. 
    \end{cases} 
\end{equation}
The utility of the information designer is
\begin{equation}
     \forall \omega \in \Omega, \quad u(a=1, \omega) = 1, \quad u(a=0, \omega) = 0, 
\end{equation}
namely, the designer always prefers that the receiver take action $1$. 
An optimal signaling scheme $\pi^*$ for this Bayesian persuasion instance is the following direct signaling scheme, with signal set $S = A = \{0, 1\}$: 
\begin{equation}
    \pi^*(1 \mid 1) = 1, \quad  \pi^*(0 \mid 1) = 0, \quad \quad \quad \pi^*(1 \mid 0) = v^*, \quad \pi^*(0 \mid 0) = 1-v^*,   
\end{equation}
which gives the optimal utility
\begin{equation}
    U^* = \prior(0)\pi(1|0) + \prior(1)\pi(1|1) = \frac{1+\eps}{1+\eps v^*}v^*. 
\end{equation}

Before further analysis, we emphasize that it is without loss of generality to assume that the designer in the above Bayesian persuasion problem only uses signaling schemes that send signal $1$ deterministically when the state is $1$. Therefore, the receiver believes for sure that the state is $0$ when receiving signal $0$, while the state is still uncertain when signal $1$ is sent.  This is formalized by Lemma \ref{lem:two-signal-deterministic}. 
\begin{lemma}
\label{lem:two-signal-deterministic}
For the above family of Bayesian persuasion problems, for any algorithm $\Alg^{\mathrm{BP}}$ of the information designer that uses signaling schemes with two signals, there exists another algorithm $\overline\Alg^{\mathrm{BP}}$ that only uses signaling schemes $\pi^{(t)}$ satisfying $\pi^{(t)}(s=1 \mid \omega=1) = 1$ with weakly smaller regret $\Reg(\overline\Alg^{\mathrm{BP}}) \le \Reg(\Alg^{\mathrm{BP}})$.  Additionally, the receiver always takes action $0$ when $\overline{\pi}^{(t)}$ sends signal $0$. 
\end{lemma}
\begin{proof}

The proof is analogous to the ``Construction of the dynamic pricing mechanism [Step 1]'' part in the proof of Lemma 3.8 in \cite{feng_online_2022} where they show that it is without loss of generality to consider signaling schemes that deterministically send signal $1$ at state $1$.  Although they assume unknown utility and know prior while we assume know utility and unknown prior, this difference does not affect the proof. Hence, we omit the details. 
\end{proof}

With a signal set $S = A = \{0, 1\}$, 
given a signaling scheme $\pi$ satisfying $\pi(s=1 \mid \omega=1) = 1$, the receiver takes action $1$ at signal $s=1$ if and only if
\begin{align*}
    & 0 ~ \le ~ \prior(0)\pi(s=1 \mid \omega = 0)v(1,0) + \prior(1)\pi(s=1 \mid \omega = 1)v(1,1) \\
    & ~ \iff ~ \pi(1|0) ~ \le ~ \frac{\prior(1) v(1, 1)}{-\prior(0) v(1, 0)} \pi(1|1) ~ = ~ \eps v^* \cdot \frac{1}{\eps} \cdot 1 ~ = ~ v^*, 
\end{align*}
and takes action 0 at signal 0. 
The expected utility of the designer is
\begin{align*}
    U(\prior, \pi) & ~ = ~ \Pr[\text{signal $1$ is sent}]\cdot \mathbbm{1}[\text{receiver takes action $1$ given signal $1$}] ~ + ~ 0\\
    & ~ = ~ \Big(\prior(0)\pi(1 | 0) + \prior(1)\pi(1 | 1) \Big) \cdot \mathbbm{1}\big[ \pi(1|0) \le v^*\big] \\
    & ~ \le ~ \pi(1|0) \cdot \mathbbm{1}\big[\pi(1|0) \le v^*\big] + \eps,  
\end{align*}
and her regret of using a signaling scheme $\pi$ in a single period is 
\begin{align}
    U^* - U(\prior, \pi) & ~ \ge ~ \tfrac{1+\eps}{1+\eps v^*}v^* - \Big(\pi(1|0) \cdot \mathbbm{1}\big[\pi(1|0) \le v^*\big] + \eps \Big) \nonumber \\
    & ~ \ge ~ v^* - \pi(1|0) \cdot \mathbbm{1}\big[\pi(1|0) \le v^*\big] - \eps. \label{eq:lower-bound-U^*-U} 
\end{align}

Given an algorithm $\Alg^{\mathrm{BP}}$ for the Bayesian persuasion problem as a subroutine, we construct an algorithm $\Alg^{\mathrm{pricing}}$ for the dynamic pricing problem that posts price $p_t = \pi^{(t)}(1|0)$: 
\begin{mdframed}[frametitle={Dynamic pricing algorithm $\Alg^{\mathrm{pricing}}$:}]
In each period $t = 1, \ldots, T$, after a buyer arrives: 

\begin{itemize}[topsep=0.3em, parsep=0em]
    \item Obtain the signaling scheme $\pi^{(t)}$ for this period from $\Alg^{\mathrm {BP}}$.
    \item Post price $p_t = \pi^{(t)}(1|0)$ to the buyer. 
    \item Regard the state $\omega^{(t)}$ as $0$. Draw signal $s^{(t)} \sim \pi^{(t)}(\cdot | \omega^{(t)}=0)$.
    \begin{itemize}
        \item If $s^{(t)} = 0$, give feedback $(\omega^{(t)}=0, s^{(t)}=0, a^{(t)}=0)$ to $\Alg^{\mathrm {BP}}$.
        \item If $s^{(t)} = 1$, give feedback $(\omega^{(t)}=0, s^{(t)}=1, a^{(t)} = I_t)$ to $\Alg^{\mathrm {BP}}$, where $I_t$ is an indicator for whether the buyer buys the item: $I_t=1$ for yes and $I_t=0$ for no. 
    \end{itemize}
\end{itemize}
\end{mdframed}

\begin{lemma}\label{lem:BP-pricing-reduction}
The regret of $\Alg^{\mathrm{pricing}}$ in the dynamic pricing problem and the regret of $\Alg^{\mathrm{BP}}$ in the Bayesian persuasion problem satisfy $\E[\Reg(\Alg^{\mathrm{pricing}})] \le \E[\Reg(\Alg^{\mathrm{BP}})] + \frac{2}{T}$. 
\end{lemma}

\begin{proof}
In the execution of $\Alg^{\mathrm {pricing}}$, $\Alg^{\mathrm{BP}}$ is faced with the Bayesian persuasion problem conditioning on the event $\mathcal E$ that all states $\omega^{(t)}$ are realized to be $0$.
The probability that all states are realized to be $0$ is at least
\begin{equation*}
    \Pr[\mathcal E] ~ \ge ~ 1 - T \prior(1) ~ = ~ 1 - T \frac{\eps v^*}{1+\eps v^*} ~ \ge ~ 1 - T\eps ~ = ~ 1 - \frac{1}{T^2}
\end{equation*}
given $\eps = \frac{1}{T^3}$ and $v^* \le 1$. 
Thus, the total expected regret of algorithm $\Alg^{\mathrm{BP}}$ on the original Bayesian persuasion problem (without conditioning on event $\mathcal E$) satisfies
\begin{align}
    \E[\Reg(\Alg^{\mathrm{BP}})] & ~ \ge ~ 
    \Pr[\mathcal E] \cdot \E[\Reg(\Alg^{\mathrm{BP}}) \mid \mathcal E]  \nonumber \\
    & ~ \ge ~ \Big(1 - \frac{1}{T^2}\Big) \cdot \E[\Reg(\Alg^{\mathrm{BP}}) \mid \mathcal E] \nonumber \\
    & ~ \ge ~  \E[\Reg(\Alg^{\mathrm{BP}}) \mid \mathcal E] - \frac{T}{T^2}  \nonumber \\
    & ~ = ~ \E\Big[ \sum_{t=1}^T \Big( U^* - U(\prior, \pi^{(t)} \Big) \mid \mathcal E \Big] - \frac{1}{T}  \nonumber \\
    \text{by \eqref{eq:lower-bound-U^*-U}} & ~ \ge ~ \E\Big[ \sum_{t=1}^T \Big( v^* - \pi^{(t)}(1|0) \cdot \mathbbm{1}\big[\pi^{(t)}(1|0) \le v^*\big] - \eps \Big) \mid \mathcal E \Big] - \frac{1}{T}  \nonumber \\
    & ~ \ge ~ \E\Big[ \sum_{t=1}^T \Big( v^* - \pi^{(t)}(1|0) \cdot \mathbbm{1}\big[\pi^{(t)}(1|0) \le v^*\big] \Big) \mid \mathcal E \Big] - \frac{2}{T}. 
    \label{eq:E[Reg-ALG-BP]-lower-bound}
\end{align}
Therefore, the regret of the dynamic pricing algorithm $\Alg^{\mathrm{pricing}}$ is 
\begin{align*}
    \E[\Reg(\Alg^{\mathrm{pricing}})] & ~ = ~ \E\Big[ Tv^* - \sum_{t=1}^T p_t \cdot \mathbbm{1}\big[p_t\le v^*\big] \Big] \\
    & ~ = ~ \E\Big[ T v^* - \sum_{t=1}^T \pi^{(t)}(1|0)\cdot \mathbbm{1}\big[\pi^{(t)}(1|0) \le v^*\big] \mid \mathcal E \Big] \\
    \text{by \eqref{eq:E[Reg-ALG-BP]-lower-bound}} & ~ \le ~ \E[\Reg(\Alg^{\mathrm{BP}})]  + \frac{2}{T}. 
\end{align*}
\end{proof}

\begin{proof}[Proof of Theorem~\ref{thm:2-action-lower-bound}]
    If there exists $\Alg^{\mathrm{BP}}$ for the Bayesian persuasion problem with regret less than $\Omega(\log \log T)$, then by Lemma~\ref{lem:BP-pricing-reduction} there exists a dynamic pricing algorithm satisfying $\E[\Reg(\Alg^{\mathrm{pricing}})] \le \E[\Reg(\Alg^{\mathrm{BP}})] + \frac{2}{T} < \Omega(\log \log T)$, which contradicts Theorem~\ref{thm:dynamic-pricing-lower-bound}. 
\end{proof}

\end{document}